\documentclass[a4paper,11pt]{article}

\pdfoutput=1

\usepackage{hyperref}
\usepackage{epsfig}
\usepackage{amsmath}
\usepackage{amsthm}
\usepackage{amssymb}
\usepackage{url}
\usepackage{comment}
\usepackage[left=25mm,right=25mm,top=25mm,bottom=30mm]{geometry}

\newtheorem{theorem}{Theorem}
\newtheorem{lemma}[theorem]{Lemma}
\newtheorem{corollary}[theorem]{Corollary}
\newtheorem{proposition}[theorem]{Proposition}
\newtheorem{observation}[theorem]{Observation}
\theoremstyle{definition}
\newtheorem*{definition}{Definition}

\newcommand{\eps}{\varepsilon}

\title{Point Sets with Many Non-Crossing Perfect Matchings}

\author{Andrei Asinowski\thanks{Institut f\"ur Informatik, Freie
    Universit\"at Berlin. E-mail \href{mailto:asinowski@inf.fu-berlin.de}
{\texttt{asinowski@inf.fu-berlin.de}}.
{Supported by the
ESF EUROCORES programme EuroGIGA-ComPoSe, Deutsche Forschungsgemeinschaft (DFG): FE 340/9-1.}},
G\"unter Rote\thanks{Institut f\"ur Informatik, Freie Universit\"at Berlin. E-mail \href{mailto:rote@inf.fu-berlin.de}{\texttt{rote@inf.fu-berlin.de}}.}}

\begin{document}

\maketitle

\begin{abstract}
The maximum number of non-crossing straight-line perfect matchings that a set of $n$ points in the plane can have is known to be $O(10.0438^n)$ and $\Omega^*(3^n)$. The lower bound, due to Garc\'ia, Noy, and Tejel (2000) is attained by the \emph{double chain},
which has $\Theta(3^n n^{O(1)})$ such matchings. We reprove this bound in a simplified way that uses the novel notion of \emph{down-free matching}, and apply this approach on several other constructions. As a result, we improve the lower bound. First we show that \emph{double zigzag chain} with $n$ points has $\Theta^*(\lambda^n)$ such matchings with $\lambda \approx 3.0532$. Next we analyze further generalizations of double zigzag chains -- \emph{double $r$-chains}. The best choice of parameters leads to a construction with $\Theta^*(\nu^n)$ matchings, with $\nu \approx 3.0930$. The derivation of this bound requires an analysis of a coupled dynamic-programming recursion between two infinite vectors.
\end{abstract}

\tableofcontents

\section{Introduction}\label{sec:intro}

\paragraph{Background.}
A \textit{non-crossing straight-line matching} of a
finite planar point  set
is a
graph
whose vertices are the given points,
whose edges are realized by pairwise non-crossing straight segments,
and where
every vertex has degree at most~$1$.
In what follows, such matchings will be simply called
\emph{matchings}.
A matching is \textit{perfect} if every point is matched --
that is, has degree~$1$.
Throughout the paper, all point sets are assumed to be
in general position in the sense that no three points lie on a line.

In this paper we deal with bounds on the number of
perfect matchings that a set of size $n$ can have.
This question arises in a broader context.
Non-crossing straight-line matchings, either perfect or not necessarily perfect, are just
two kinds of {geometric plane graphs}, others being triangulations,
spanning trees, connected graphs, etc.
A web page of Adam Sheffer
\footnote{Numbers of Plane Graphs.
\url{http://www.cs.tau.ac.il/~sheffera/counting/PlaneGraphs.html}}
maintains the best up-to-date bounds on the maximum number of geometric plane graphs of several kinds.

First we recall that for the \emph{minimum} number of perfect
matchings that $n$ points in general position can have,
the exact solution is known.
Garc\'ia, Noy, and Tejel~\cite{garcia} proved
the number of perfect matching is minimized on point sets in convex position.
It is well-known that the number of perfect matchings is then
$C_{n/2}$, where $C_k = \frac{1}{k+1}\binom{2k}{k}
= \Theta(4^k/k^{3/2})$ is the $k$-th Catalan number.
 The minimum number $C_{n/2}$  of perfect matchings is in fact attained
\emph{only} for point sets in convex position, with the exception of one configuration of six points~\cite{aa}.

Regarding the \emph{maximum} number of perfect matchings that a point set of size $n$ can have,
only asymptotic bounds are known.
The best upper bound to date,
$O(10.0438^n)$, was proved by Sharir and Welzl~\cite{sw}.
The best previous lower bound was given by Garc\'ia, Noy, and Tejel in the above-mentioned paper~\cite{garcia}.
They showed that for the so-called \emph{double chain} with $n$ points (to be denoted by  $\mathrm{DC}_n$)
the following holds:

\begin{theorem} [{\cite[Theorem 2.1]{garcia}}]\label{thm:gnt}
The number of perfect matchings of the double chain with $n$ points is
$\Theta\left(3^n n^{O(1)}\right)$.
\end{theorem}

Actually, it follows from their proof that this number is $\Omega(3^n / n^4)$ and $O(3^n / n^3)$.
In Section~\ref{sec:dc_pm} we shall sketch this proof,
and also determine the polynomial factor more precisely.

The double chain was used in~\cite{garcia} not only for improving the lower bounds on the maximum number
of perfect matchings, but also for some
other kinds of geometric graphs: triangulations, spanning trees and polygonizations.
 It was believed by some researchers in the field that
it might give the true upper bound at least for some of these kinds.
However, in 2006,
Aichholzer, Hackl, Huemer, Hurtado, Krasser, and Vogtenhuber~\cite{aich1}
introduced a new construction, the \emph{double zigzag chain} with $n$ points
(to be denoted by $\mathrm{DZZC}_n$), see Figure~\ref{fig:dzzc} below.
They proved that
$\mathrm{DZZC}_n$ improves the lower bound for
the number of triangulations: it is
$\Theta^*(8^n)$ for $\mathrm{DC}_n$ and $\Theta^*(8.48^n)$ for $\mathrm{DZZC}_n$.
(The notations $O^*$, $\Omega^*$, and $\Theta^*$
correspond to the usual $O$-, $\Omega$- and $\Theta$-notations, but
with polynomial factors omitted.)
To our knowledge, the number of geometric graphs of other kinds mentioned above
for $\mathrm{DZZC}_n$ was not found.

In this paper we determine asymptotically the number of perfect matchings for
$\mathrm{DZZC}_n$ and its further generalizations, improving the existing lower bound.

\paragraph{Our results.}

In Section~\ref{zigzag-chains}, we will first show that $\mathrm{DZZC}_n$ has asymptotically more perfect matchings than~$\mathrm{DC}_n$:
\begin{theorem}\label{thm:main}
The number of perfect matchings of the double zigzag chain with $n$ points
is $\Theta^*(\lambda^n)$,
where $\lambda = \sqrt{
  (\sqrt{93}+ 9)/2
} \approx 3.0532$.
\end{theorem}

 In Sections~\ref{sec:r_chains} and~\ref{sec:with},
we will present a generalization of $\mathrm{DC}_n$, which comes
in two variations:
\emph{double $r$-chains without corners} and \emph{double $r$-chains with corners},
see Figures~\ref{fig:ch_with_corners}
and~\ref{fig:ch_without_corners} below.
Our best results for these constructions are as follows:

\begin{theorem}\label{thm:main_wo_corners}
The number of perfect matchings of the
double $11$-chain without corners with $n$ points
is $\Theta^*(\nu^n)$,
where $\nu = \sqrt[11]{240054} \approx 3.0840$.
\end{theorem}

\begin{theorem}\label{thm:main_corners}
The number of perfect matchings of the
double $8$-chain with corners with $n$ points
is $\Omega((\nu-\varepsilon)^n)$,
and $O(\nu^n)$,
where $\nu
=\sqrt[8]{\bigl(8389 + 3\sqrt{7771737}\,\bigr)\big/2}
\approx 3.0930 $
and $\varepsilon>0$ is arbitrarily small.
\end{theorem}

We shall present proofs for all three theorems because they use different techniques.
First, in Section~\ref{sec:dzzc_pm} we introduce the notion of \emph{down-free matchings}
and show in Theorem~\ref{thm:double} how one can generally reduce the problem of
asymptotic enumeration
of perfect matchings of a ``double structure''
to that
of down-free matchings of the corresponding ``single structure''.
In the proof of Theorem~\ref{thm:main} (Section~\ref{zigzag-chains}),
we find a recursion for the number of down-free matchings of the zigzag chain,
and translate it into a functional equation satisfied by the generating function.
 We solve this equation explicitly, which allows us to find the
 asymptotic growth rate by looking at the smallest singularity of the function.
In the proof of Theorem~\ref{thm:main_wo_corners}
(Section~\ref{sec:r_chains})
we use matchings which possibly have \emph{runners} --
edges with only one endpoint assigned.
We define a sequence of infinite vectors whose entries are
the numbers of down-free matchings of the $r$-chain 
 of a certain size,
sorted by the number of runners.
These vectors can be computed recursively.
We reformulate this  recursion in term of lattice paths
and obtain the desired growth rate
with the help of a result of Banderier and Flajolet~\cite{ban}.
The proof of Theorem~\ref{thm:main_corners}
(Section~\ref{sec:with}) starts similarly,
but due to technical obstacles,
we need \emph{two} sequences of infinite vectors,
defined by a coupled recursion. 
We find that the desired growth constant is determined
by the dominant eigenvalue of certain $2 \times 2$ matrix.

\paragraph{Notation.}

We use the following notation and convention.
A \emph{construction} $X$ is a family $\{X_n\}_{n\in I}$ for some
infinite $I \subseteq \mathbb{N}$,
where, for fixed $n$, $X_n$
is a class of point sets of size $n$
with certain common properties, for example,
a certain order type (or, in some cases: one of several order types)
and certain restrictions concerning
position in the plane with respect to coordinate axes.
The double chain ($\mathrm{DC}$) mentioned above is one of such constructions.
Occasionally we will abuse notation and denote by
$X_n$ not only such a class, but also any of its representatives.
If we know that all members of $X_n$ have, for example, the same number of matchings,
we can speak unambiguously about ``the number of matchings of $X_n$'', and so on.

In what follows,
$\mathsf{pm}(X_n)$ denotes the number of perfect matchings of $X_n$;
$\mathsf{am}(X_n)$, the number of all (non-crossing straight-line, but not necessarily perfect) matchings of $X_n$;
$\mathsf{dfm}(X_n)$, the number of down-free matchings of $X_n$.
For some constructions it can happen that
not all representatives of $X_n$ have the same number of (for example) perfect matchings
and, thus, $\mathsf{pm}(X_n)$ is not well-defined,
but the common asymptotic bound still can be given, which enables us to
write expressions like $\mathsf{pm}(X_n) = \Theta^*(\mu^n)$ in such cases as well.

For two distinct points $p$ and $q$,
the straight line through $p$ and $q$ will be denoted by~$\ell(p, q)$.

A set of points is \emph{in downward position}
(respectively, \emph{in upward position}) if the points lie on
the graph of a convex (respectively, concave) function.
 In particular,
three points with different $x$-coordinates
are in downward position
(respectively, {in upward position})
if
they form a counterclockwise (respectively, clockwise)
oriented triangle when sorted by $x$-coordinate.

A point of $X$ not matched by
a
matching
will be called
a \emph{free point}.

\section{Double chains and double zigzag chains}\label{sec:chains}

In this section we recall the definitions of a double chain
and a double zigzag chain,
and recall how the bound
$\mathsf{pm}(\mathrm{DC}_n) =\Theta^*(3^n n^{O(1)})$
from Theorem~\ref{thm:gnt}
for was obtained in~\cite{garcia}.

\subsection{One set high above another and ``double constructions''}
We start with a notion of
a point set being one ``high above'' another point set.

\begin{definition}
Let $P$ and $Q$ be two point sets in the plane, each of them having distinct $x$-coordinates.
We say that $P$ is \emph{high above} $Q$ if
any point of $P$ lies above any line through two points of $Q$
and
any point of $Q$ lies below any line through two points of $P$.
\end{definition}

\begin{observation}\label{obs:high}
For any two point sets $P$ and $Q$ with distinct $x$-coordinates
it is possible to put a translate of $P$ high above a translate of $Q$.
\end{observation}

\begin{proof}
Assume that $Q$ lies in the interior of a ``topless rectangle''
$R=[a,b] \times [c, +\infty)$.
Any line that passes through two points of $Q$ crosses the lines $x=a$ and $x=b$ in certain points.
Therefore there exists a number $d$ such that
all points of $R$ above $y=d$
are also above any line containing two points of $Q$.
A similar argument (with ``below'' instead ``above'') applies for $P$.
Thus, if we put translates of $P$ and $Q$ between vertical lines
$x=a$ and $x=b$
(where $a$ and $b$ are chosen to fit bot sets)
and, if needed, translate $P$ upwards,
we get eventually a translate of~$P$ high above a translate of $Q$.
\end{proof}

Let $X_n$ be a construction.
A \emph{double $X_n$} (denoted by $\mathrm{D}X_{2n}$) is
the family of sets obtained
by placing $P$, a representative of $X_n$,
high above $Q$, a
representative of $X_n$
reflected across a horizontal line.
An edge between a point of $P$ and a point of $Q$ will be called a \emph{$PQ$-edge}.

\subsection{Double chains}\label{sec:dc}
A (single) \emph{downward chain} (respectively, \emph{upward chain}) of size $n$
is a set of $n$ points in downward (respectively, downward) position.
A downward chain of size $n$ will be denoted by $\mathrm{SC}_n$.

Let $n$ be an even number.
A \emph{double chain} of size $n$ consists of
a downward chain of size $n/2$, $P=\{p_1, p_2, \dots, p_{n/2}\}$,
placed high above
an upward chain of size $n/2$, $Q=\{q_1, q_2, \dots, q_{n/2}\}$.
See Figure~\ref{fig:dc} for an example.
A double chain of size $n$
will be denoted by $\mathrm{DC}_n$.
We assume that both sets $P$ and $Q$ are sorted by $x$-coordinate.

\begin{figure}[ht]
\centering
\includegraphics[scale=1]{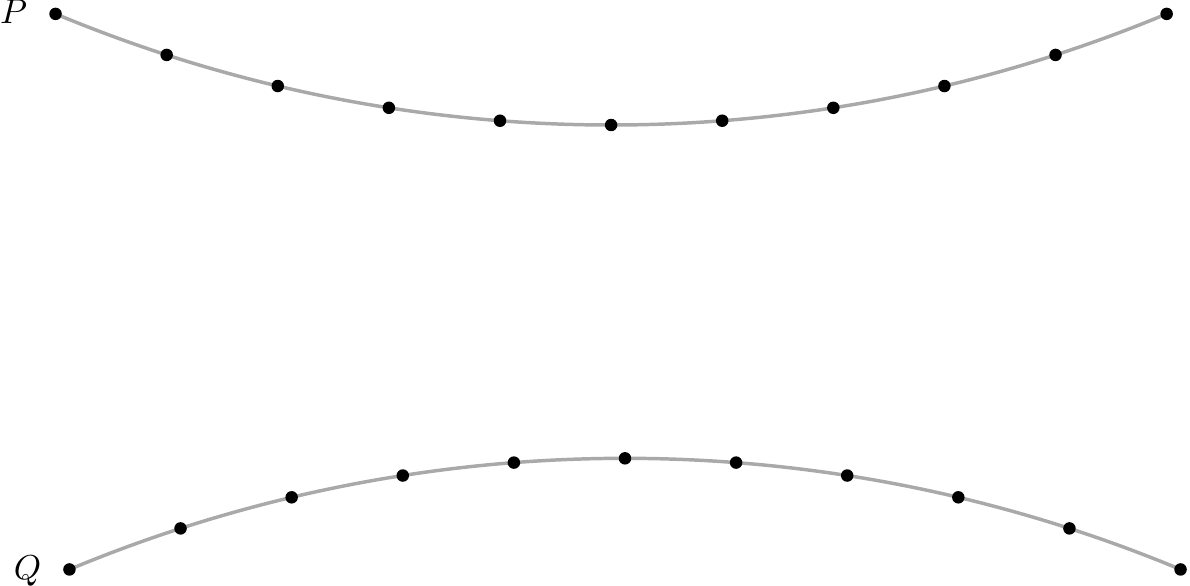}
\caption{A double chain of size $22$.}
\label{fig:dc}
\end{figure}

\subsection{Perfect matchings in the double chain}\label{sec:dc_pm}
Theorem~\ref{thm:gnt} was proved in~\cite{garcia} as follows.

Denote by $\mathsf{pm}(\mathrm{ DC}_{n, j})$ the number of perfect matchings of $\mathrm{ DC}_n$
that have exactly $j$ $PQ$-edges between the upper and the lower chain.
If $n/2-j$ is odd, then no perfect matching exists,
so we assume that $n/2-j$ is even.
One can construct a perfect matching with $j$ $PQ$-edges
in the following way.
First choose any $j$ points of $P$ and $j$ points of $Q$ and
connect them by $j$ non-intersecting $PQ$-edges.
It is easy to see that
there is a unique way to connect the chosen points
(see also Proposition~\ref{prop:dw-highabove} below).
 Then, choose any perfect matching of the
 free points in each chain.
 Alternatively, one can first choose $n/2-j$ points of $P$ and $n/2-j$ points of $Q$,
then take any matching of $P$ and any matching of $Q$ that uses the chosen points;
after that, the free points
can be matched by $PQ$-edges in a unique way.
Since $Q$ has the same order type as $P$,
it follows that
\begin{equation}\label{eq:dc1}
\mathsf{ pm}(\mathrm{ DC}_{n, j}) =
\left(  \mathsf{ am}(\mathrm{ SC}_{n/2, j})  \right)^2=
 \left( \binom{n/2}{j} \cdot C_{(n/2-j)/2} \right)^2,
 \end{equation}
 where $\mathsf{ am}(\mathrm{ SC}_{n/2, j})$ denotes the number of matchings of $P$
 (or, equivalently, of any set of $n/2$ points in convex position) with exactly $j$ free points.
Finally, the total number of perfect matchings of $\mathrm{ DC}_n$
is
\begin{equation}\label{eq:dc2}
\mathsf{ pm}(\mathrm{ DC}_{n}) =
\sum_{
\substack{
  0 \leq j \leq n/2 \\
j\equiv n/2 \pmod 2
}
}\left( \binom{n/2}{j} \cdot C_{(n/2-j)/2} \right)^2.
\end{equation}
An analysis shows that the dominant term in this sum
is the term corresponding to $j\approx n/6$,
that is
$\left(\binom{n/2}{n/6} \cdot C_{n/6} \right)^2$
($n/6$ should be rounded in one way or the other).
Using the estimates
$\binom{ak}{bk} = \Theta{\left(\left(\frac{a^a}{b^b \ (a-b)^{a-b}}\right)^k / k^{1/2}\right)}$
and $C_k = \Theta(4^k/k^{3/2})$,
which follow from Stirling's formula,
one obtains
$\mathsf{ pm}(\mathrm{DC}_{n, n/6})
=
\Theta(3^n/n^4)$,
and, therefore, $\mathsf{ pm}(\mathrm{ DC}_{n})=\Omega(3^n/n^4)$
and $O (3^n/n^3)$.
With the help of Stirling's formula, and replacing the sum \eqref{eq:dc2}
by an integral, one can obtain the more precise estimate
$\mathsf{ pm}(\mathrm{DC}_{n})=3^n/n^{7/2}\cdot182/\pi^{3/2}\cdot(1+o(1))$
(we omit the details).

\subsection{Double zigzag chains}\label{sec:dzzc}

In this section we recall the definitions of a (single) zigzag chain $\mathrm{SZZC}$
and a double zigzag chain $\mathrm{DZZC}$.

Let $P = \{p_1, p_2, \dots, p_n\}$ be a downward chain ($\mathrm{SC}_n$) sorted by $x$-coordinate.
For each even~$i$, $1<i<n$, we move the point $p_i$ vertically up,
very slightly above the segment $p_{i-1}p_{i+1}$,
so that all triples $p_{i-1}p_{i}p_{i+1}$ with even $i$ ($1<i<n$) are now in upward position,
and all other triples of points remain in downward position.
After this modification, the points $p_1, p_2, \dots, p_{n}$ are still sorted by $x$-coordinate.
A set obtained in this way will be called
an \emph{even (single) downward zigzag chain} of size $n$
and denoted by $\mathrm{eSZZC}_n$.
If instead of even $i$-s we perform this transformation for each odd $i$, $1<i<n$,
we obtain an \emph{odd (single) downward zigzag chain} ($\mathrm{oSZZC}_n$).
If $n$ is even, then $\mathrm{eSZZC}_n$ and $\mathrm{oSZZC}_n$
are reflections of each other with respect to a vertical line;
but if $n$ is odd, then $\mathrm{eSZZC}_n$ and $\mathrm{oSZZC}_n$
have different order types, and -- as one can verify on some small examples --
different numbers of (perfect or not necessarily perfect) matchings.
See Figure~\ref{fig:zz} for an example.
A zigzag chain of size $n$, denoted by $\mathrm{SZZC}_n$,
is either an $\mathrm{eSZZC}_n$ or an $\mathrm{oSZZC}_n$.
 For both types of $\mathrm{SZZC}_n$,
 we shall derive the same asymptotic bound on the number of perfect
 matchings.

\begin{figure}[ht]
\centering
\includegraphics[scale=0.814]{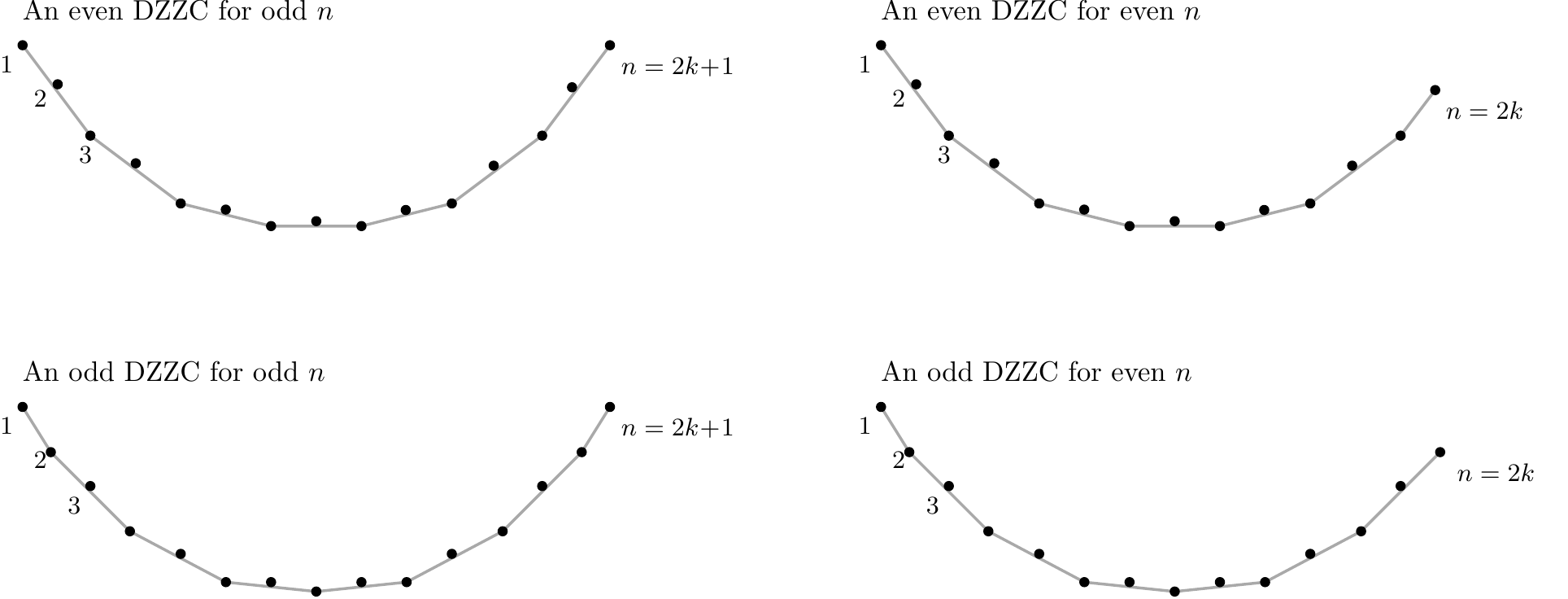}
\caption{A (single) zigzag chain -- several cases.}
\label{fig:zz}
\end{figure}


An \emph{upward zigzag chain} (of either kind)
is a downward zigzag chain reflected across a horizontal line.
The construction of
a double zigzag chain from zigzag chains is analogous to
the construction of the double chain from two single chains:
A \emph{double zigzag chain} of (even) size~$n$ ($\mathrm{ DZZC}_n$)
consists of
a downward zigzag chain $P=\{p_1, p_2, \dots, p_{n/2}\}$
high above
an upward zigzag chain $Q=\{q_1, q_2, \dots, q_{n/2}\}$.
We can combine even and odd zigzag chains in various ways,
but as mentioned above, this will make no difference for the
asymptotic number of perfect matchings.
See Figure~\ref{fig:dzzc} for an example
of double zigzag chain obtained from two even zigzag chains of odd size.

\begin{figure}[ht]
\centering
\includegraphics[scale=1]{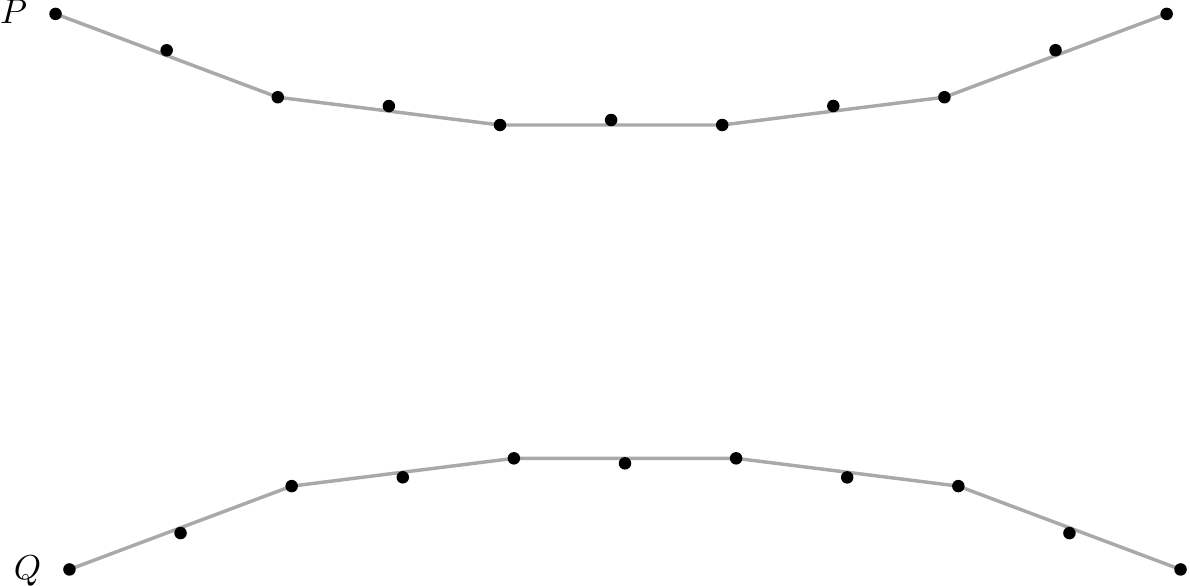}
\caption{A double zigzag chain of size $n=22$.}
\label{fig:dzzc}
\end{figure}

\section{Down-free matchings and perfect matchings}\label{sec:dzzc_pm}
\subsection{Down-free matchings}
Suppose that we want to adapt
the argument that was used for estimating
$\mathsf{ pm}(\mathrm{ DC}_{n})$
for the case of $\mathsf{ pm}(\mathrm{DZZC}_{n})$ (of any kind).
That is: for fixed $j$ (such that $n/2-j$ is even) we want to choose $j$ $PQ$-edges,
and to complete this matching to a perfect matching
by choosing edges that connect free points of the same chain in all possible ways.
One can hope for improvement
since
the number of perfect matchings in $\mathrm{SZZC}_n$ (of any kind)
is $\Theta^*(\nu^n)$
with $\nu = \sqrt{2+2\sqrt{2}} \approx 2.1974$, in contrast to $\Theta^*(2^n)$ for $\mathrm{SC}_n$.
(This bound for $\mathrm{SZZC}_n$
was proven in~\cite{aich1} for a slightly different construction, the so called \emph{double circle}.
The order type of $\mathrm{SZZC}$ is different from that of a double circle only in one triple of points;
it is easy to show that they have the same asymptotic number of perfect matchings.)
However, in comparison with the case of $\mathrm{DC}_n$,
here we have less freedom and no uniformity in constructing the matchings inside $P$ and $Q$,
once $PQ$-edges are chosen.
Indeed, the $j$ chosen $PQ$-edges may block visibility between certain pairs of
free points from $P$ or from $Q$.
Moreover,
for different choices of $j$ $PQ$-edges,
we have in general different numbers of ways
to complete them to a perfect matching of  $\mathrm{DZZC}_n$.
This follows from the fact that sets of points that remain free after choosing $j$ $PQ$-edges have in general
various order types, and, so, it seems hopeless to enumerate them in this way.
On the other hand, if we first choose
$(n-2j) / 4$
 edges between two points of $P$
and $(n-2j) / 4$
 edges between two points of~$Q$,
then -- as we prove below in Proposition~\ref{prop:dw-highabove} --
there is \emph{at most} one way to complete such a matching
to a perfect matching of  $\mathrm{DZZC}_n$.
More precisely, if the free points of $P$ ``see'' all free points of $Q$,
there is exactly one way of complete a matching to a perfect one,
otherwise it is impossible.
Next we define a property of matchings
which -- for two sets being one high above another --
ensures the desired visibility of free points.

\begin{definition}
Let $P$ be a set of points with distinct $x$-coordinates.
Let $M$ be a matching of~$P$.
$M$ is a \emph{down-free matching} if
for each unmatched point $p \in P$,
no edge of $M$ has a point directly below $p$.
In other words: for each free point $p \in P$, the vertical ray going down from $p$,
does not cross any edge of $M$.
Similarly, one defines an \emph{up-free matching}.
\end{definition}

\begin{proposition}\label{prop:dw-highabove}
Let $P$ and $Q$ be two point sets in general position with distinct $x$-coordinates
such that
$P$ is high above $Q$.
Let $M_P$ be a
matching of $P$ and $M_Q$ a
matching of $Q$,
such that $M_P$ and $M_Q$ have the same number of free points.
\begin{enumerate}
\item If $M_P$ is a down-free matching and $M_Q$ is an up-free matching,
then $M_P \cup M_Q$ can be completed to a perfect matching of $P \cup Q$ in a unique way.
\item If $M_P$ is not down-free or $M_Q$ is not up-free,
then it is impossible to complete $M_P \cup M_Q$ to a perfect matching of $P \cup Q$.
\end{enumerate}

\end{proposition}

\begin{proof}
We assume again that the points $P=\{p_1, p_2, p_3, \dots\}$ and $Q=\{q_1, q_2, q_3, \dots\}$ are sorted by $x$-coordinate.

First we observe that
for any $p_\alpha, p_\beta \in P$,
$q_\gamma, q_\delta \in Q$,
the points $p_\alpha, p_\beta ,q_\gamma, q_\delta $ are in convex position.
Indeed, if, for example, $q_\delta  \in \mathrm{conv}(p_\alpha, p_\beta ,q_\gamma)$,
then the points $p_\alpha$ and $p_\beta$ would lie on different sides of the line $\ell(q_\gamma, q_\delta)$,
and this contradicts $P$ being high above $Q$.

Now we show that for $\alpha <\beta$ and $\gamma < \delta$
the points $p_\alpha, p_\beta ,q_\delta, q_\gamma $ lie on the
boundary of their convex hull in this
clockwise order,
see Figure~\ref{fig:highabove}(a) for an illustration.  Since $P$ lies
high above $Q$, the points of $Q$ lie below $\ell(p_\alpha, p_\beta)$
and thus the points $p_\alpha$ and $p_\beta$ lie on the convex hull
consecutively and in this clockwise order.
Similarly $q_\delta$ and $q_\gamma$ lie on the convex hull
consecutively and in this clockwise order.
 This implies the claim.

\begin{figure}[htb]
\centering
\includegraphics[scale=1]{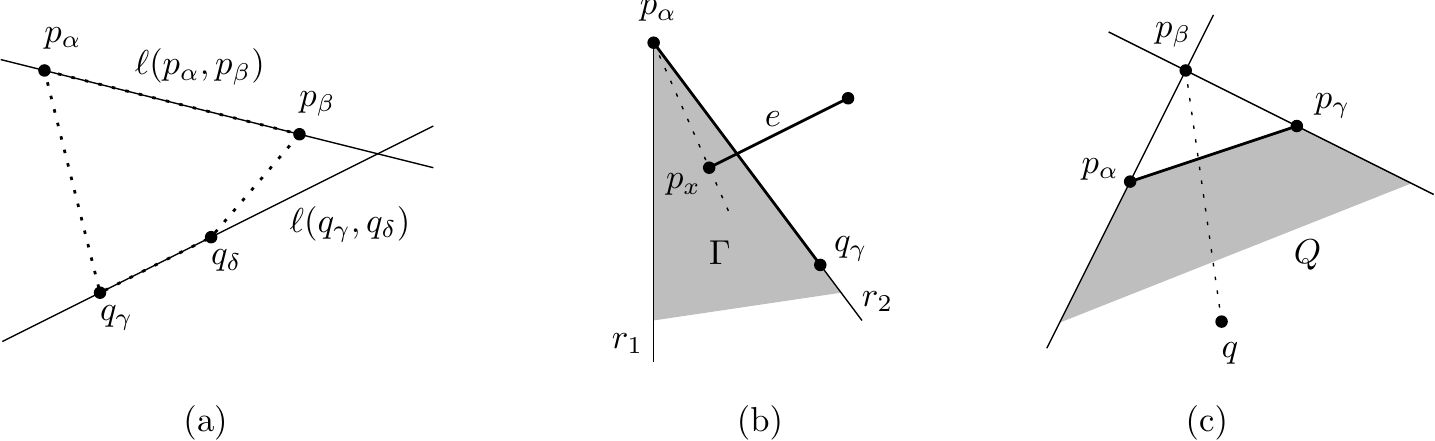}
\caption{Illustrations to the proof of Proposition~\ref{prop:dw-highabove}.}
\label{fig:highabove}
\end{figure}

\begin{enumerate}
\item Assume that $M_P$ is down-free and $M_Q$ is up-free.

Let $p_{\alpha_1}, p_{\alpha_2}, \dots, p_{\alpha_m}$ be the free points of $P$
and let $q_{\gamma_1}, q_{\gamma_2}, \dots, q_{\gamma_m}$ be the free points of $Q$,
sorted from left to right. 
We complete $M_P \cup M_Q$ to a perfect matching of $P \cup Q$
by connecting $p_{\alpha_i}$ with $q_{\gamma_i}$ for $i=1, 2, \dots, m$.
By the just-proven claim about the cyclic order of $p_\alpha, p_\beta ,q_\delta, q_\gamma $, these new $PQ$-edges do not cross each other.
Moreover, they do not cross the edges of $M_P$ and of $M_Q$.
Indeed, assume that an edge $p_{\alpha} q_{\gamma} =
p_{\alpha_i} q_{\gamma_i}$
crosses an edge $e \in M_P$.
Consider the angular sector $\Gamma$
bounded by the downward vertical ray $r_1$ with the origin
$p_{\alpha}$
and
the ray $r_2$ from
$p_{\alpha}$
through~
$q_{\gamma}$,
see Figure~\ref{fig:highabove}(b).
The edge $e$
crosses the ray~$r_2$ by assumption
and does not cross the ray $r_1$, because the matching $M_P$ is down-free.
Therefore, one of the endpoints of $e$, say $p_x$, lies in the interior of $\Gamma$.
However, this is impossible because in such a case
$q_{\gamma}$
lies above the line
$\ell(p_{\alpha}, p_x)$.

Finally, we need to show that this is the only way to complete $M_P \cup M_Q$
to a perfect matching of $P \cup Q$.
Indeed, for any other possibility to match the free points we
would have a pair of edges
$p_\alpha q_\delta$ and
$p_\beta q_\gamma$
with $\alpha < \beta$, $ \gamma < \delta $.
However, it follows from the claim
about the cyclic order of $p_\alpha, p_\beta ,q_\delta, q_\gamma $
that such edges necessarily cross.

\item Assume without loss of generality that $M_P$ is not down-free.
Then there is
 a free point $p_\beta$ in $M_P$
so that the vertical downward ray from $p_\beta$ crosses
an edge
 $p_\alpha p_\gamma$,
with 
$\alpha < \beta < \gamma$.
See Figure~\ref{fig:highabove}(c) for an illustration.
The set $Q$ must lie below
 $\ell(p_\alpha, p_\beta)$,
 $\ell(p_\alpha, p_\gamma)$, and
 $\ell(p_\beta, p_\gamma)$. There is no way to connect
 $p_\beta$ to a point $q\in Q$ without crossing the edge~$p_\alpha p_\gamma$.
\qedhere
\end{enumerate}
\end{proof}

\subsection{Down-free matchings of $X$ and perfect matchings of
  double~$X$}
\label{double-X}
In the following theorem we show how
asymptotic bounds on $\mathsf{dfm}$ for a structure $X$
imply those on $\mathsf{pm}$ for the corresponding double structure $\mathrm{D}X$.

\begin{theorem}\label{thm:double}
Let $X$ be a 
construction
so that $\mathsf{dfm}(X_n) = \Theta^*(\lambda^n)$.
Then for the double structure $\mathrm{D}X$ we have
$\mathsf{pm}(\mathrm{D}X_n) = \Theta^*(\lambda^n)$.

More precisely:
If $\mathsf{dfm}(X_n) = \Theta(\lambda^n / n^\alpha)$,
then
$\mathsf{pm}(\mathrm{D}X_n) = \Omega(\lambda^n / n^{2\alpha+1})$ and $O(\lambda^n / n^{2 \alpha})$.
\end{theorem}

\begin{proof}
Denote by $\mathsf{dfm}_j(X_{n/2})$
the number of down-free matchings of $X_{n/2}$ with exactly $j$ free points,
for $0 \leq j \leq n/2$,
and let
$p_j = \mathsf{dfm}_j(X_{n/2}) / \mathsf{dfm}(X_{n/2})$.
Then we have $\sum_{0 \leq j \leq n/2} p_j = 1$,
which implies $\frac{1}{n/2+1} \leq \sum_{0 \leq j \leq n/2} p^2_j \leq 1$.
Now
\[\mathsf{pm}(\mathrm{D}X_{n})
=
\sum_{0 \leq j \leq n/2} \mathsf{pm}_j(\mathrm{D}X_{n})
=
\sum_{0 \leq j \leq n/2} 
\mathsf{dfm}_j(X_{n/2}) 
^2
=
\mathsf{dfm}(X_{n/2})
^2 \cdot \sum_{0 \leq j \leq n/2} p_j^2,
 \]
 which implies the claim immediately.
\end{proof}

As the first application of Theorem~\ref{thm:double}, we show how one can reprove Theorem~\ref{thm:gnt}
without need to determine the dominant term in Equation~\eqref{eq:dc2}.
We use the following well-known fact.
\begin{proposition}[\cite{oeis} A001006]\label{fact:motzkin}
The number of all matchings in a set of $n$ points in convex position is the
$n$th Motzkin number $M_n$. Asymptotically, $M_n = \Theta ( 3^n/n^{3/2})$.
\end{proposition}
Moreover, every matching of a downward chain is obviously down-free.
Therefore, Theorem~\ref{thm:double}, with $\lambda=3$ and $\alpha=3/2$ gives immediately
$\mathsf{pm}(\mathrm{DC}_n) = \Omega(3^n / n^{4})$ and $O(3^n / n^{3})$.

\medskip

In the next sections we use Theorem~\ref{thm:double} for estimating $\mathsf{pm}$
for other constructions.

\section{Zigzag chains}
\label{zigzag-chains}
\label{SZZC}

By Theorem~\ref{thm:double},
asymptotic bounds on $\mathsf{dfm}(\mathrm{SZZC}_n)$
imply those on $\mathsf{pm}(\mathrm{DZZC}_n)$.
Thus, we analyze the number of down-free matchings of $\mathrm{SZZC}_n$.
We defined above two kinds of double chains: even and odd.
We introduce three generating functions depending on the kind of chain and on the parity of $n$:
\begin{enumerate}
\item $A(x)=\sum_{k\geq 0}a_kx^k$, where
$a_k = \mathsf{dfm}(\mathrm{eSZZC}_{2k+1})$;
\item $B(x)=\sum_{k\geq 0}b_kx^k$, where
$b_k = \mathsf{dfm}(\mathrm{oSZZC}_{2k+1})$;
\item $C(x)=\sum_{k\geq 0}c_kx^k$, where
$c_k = \mathsf{dfm}(\mathrm{eSZZC}_{2k})= \mathsf{dfm}(\mathrm{oSZZC}_{2k})$.
\end{enumerate}
We find recursive relationships between the coefficients of these functions.

\paragraph{Recursion for $a_k$.}
For every $k\geq 0$ we have the following cases (Figure~\ref{fig:a_cases}).
\begin{enumerate}
\item $p_1$ is not matched. This contributes $c_k$ matchings.
\item $p_1$ is matched to $p_{2i+1}$ with $2 \leq i \leq k$.
This contributes $\sum_{2 \leq i \leq k } b_{i-1} c_{k-i}$ matchings.
\item $p_1$ is matched to $p_{2i}$ with $1 \leq i \leq k$, $p_{2i-1}$ and $p_{2i+1}$ are not matched to each other.
This contributes $\sum_{1 \leq i \leq k } c_{i-1} a_{k-i}$ matchings.
\item $p_1$ is matched to $p_{2i}$ with $2 \leq i \leq k$, $p_{2i-1}$ and $p_{2i+1}$ are matched to each other.
This contributes $\sum_{2 \leq i \leq k } b_{i-2} c_{k-i}$ matchings.
\item $p_1$ is matched to $p_3$. Then $p_2$ must be matched to some point $p_{2i+1}$ with $2 \leq i \leq k$.
This contributes $\sum_{2 \leq i \leq k } b_{i-2} c_{k-i}$ matchings.
\item $p_1$ is matched to $p_3$, $p_2$ is matched to $p_{2i}$ with $2 \leq i \leq k$,
and $p_{2i-1}$ and $p_{2i+1}$ are not matched to each other.
This contributes $\sum_{2 \leq i \leq k } c_{i-2} a_{k-i}$ matchings.
\item $p_1$ is matched to $p_3$, $p_2$ is matched to $p_{2i}$ with $3 \leq i \leq k$,
and $p_{2i-1}$ and $p_{2i+1}$ are matched to each other.
This contributes $\sum_{3 \leq i \leq k } b_{i-3} c_{k-i}$ matchings.
\end{enumerate}
\begin{figure}[htb]
\centering
\includegraphics[scale=0.95]{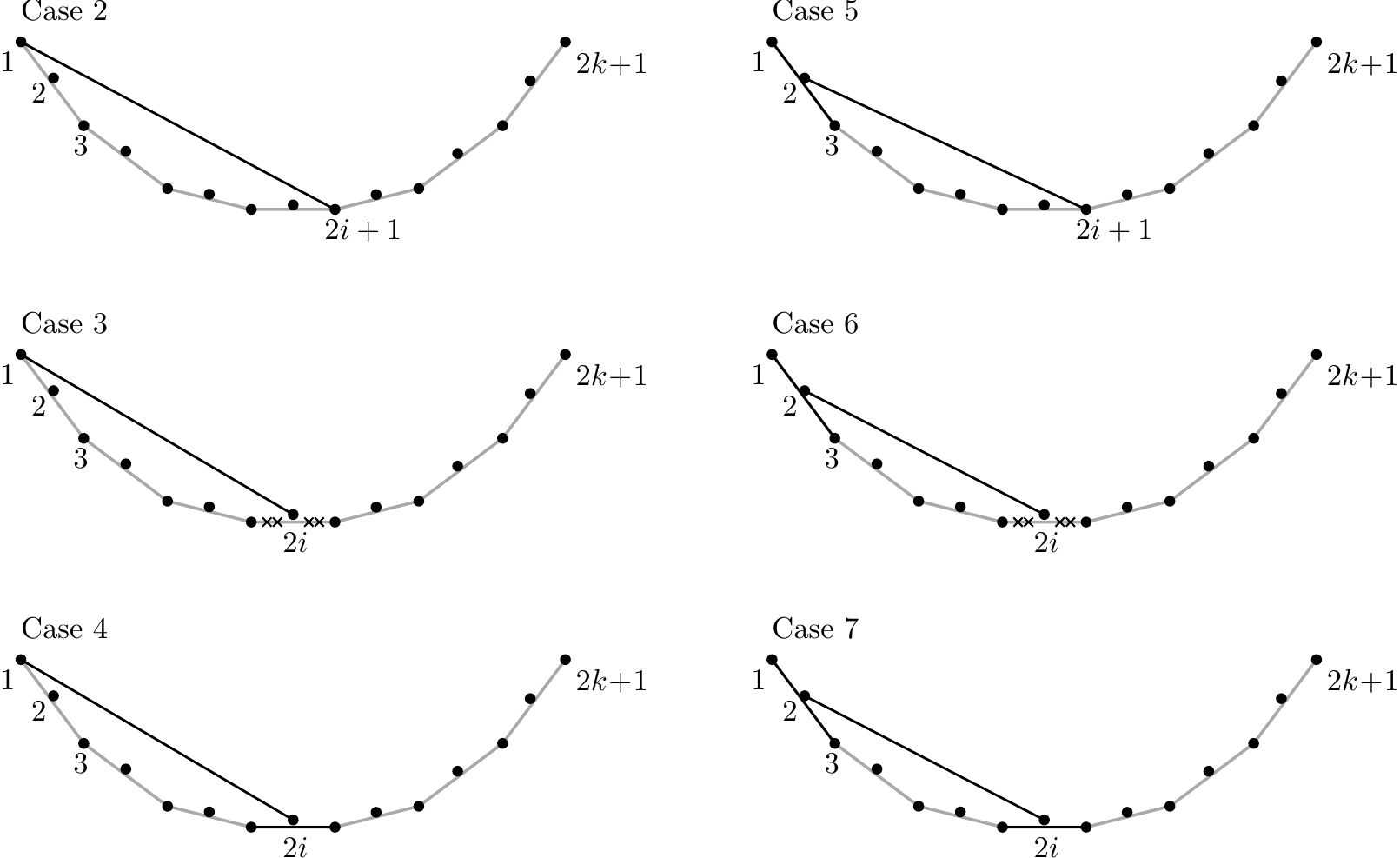}
\caption{The cases in the recursion for $a_k$.}
\label{fig:a_cases}
\end{figure}
Thus we obtain
\begin{equation}\label{eq:a}
a_k = c_k + \sum_{2 \leq i \leq k } b_{i-1} c_{k-i} + \sum_{1 \leq i \leq k } c_{i-1} a_{k-i} + 2\sum_{2 \leq i \leq k } b_{i-2} c_{k-i}+
\sum_{2 \leq i \leq k } c_{i-2} a_{k-i} + \sum_{3 \leq i \leq k } b_{i-3} c_{k-i}.
\end{equation}

\paragraph{Recursion for $b_k$.}
For every $k\geq 0$ we have the following cases, see Figure~\ref{fig:bc_cases}, left side.
\begin{enumerate}
\item $p_1$ is not matched. This contributes $c_k$ matchings.
\item $p_1$ is matched to $p_{2i}$ with $1 \leq i \leq k$.
This contributes $\sum_{1 \leq i \leq k } c_{i-1} b_{k-i}$ matchings.
\item $p_1$ is matched to $p_{2i+1}$ with $1 \leq i \leq k$, $p_{2i}$ and $p_{2i+2}$ are not matched to each other.
This contributes $\sum_{1 \leq i \leq k } a_{i-1} c_{k-i}$ matchings.
\item $p_1$ is matched to $p_{2i+1}$ with $1 \leq i \leq k-1$, $p_{2i}$ and $p_{2i+2}$ are matched to each other.
This contributes $\sum_{1 \leq i \leq k-1 } c_{i-1} b_{k-i-1}$ matchings.
\end{enumerate}
\begin{figure}[htb]
\centering
\includegraphics[scale=0.95]{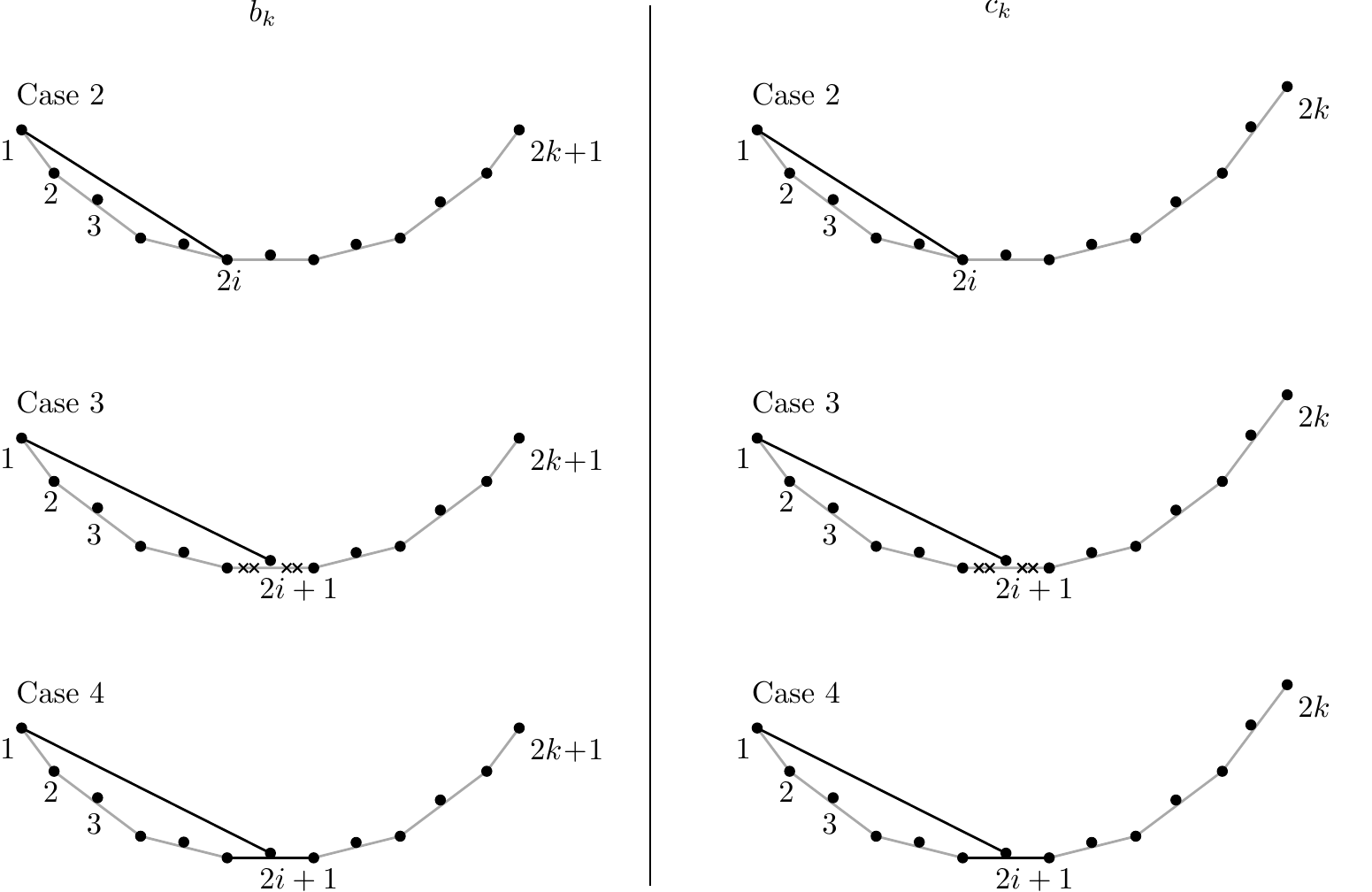}
\caption{The cases in the recursions for $b_k$ and $c_k$.}
\label{fig:bc_cases}
\end{figure}
This yields
\begin{equation}\label{eq:b}
b_k = c_k + \sum_{1 \leq i \leq k } c_{i-1} b_{k-i} + \sum_{1 \leq i \leq k } a_{i-1} c_{k-i} + \sum_{1 \leq i \leq k-1 } c_{i-1} b_{k-i-1}.
\end{equation}

\paragraph{Recursion for $c_k$.}
Clearly, $c_0=1$. For $k\geq 1$ we have the following cases, see Figure~\ref{fig:bc_cases}, right side.
\begin{enumerate}
\item $p_1$ is not matched. This contributes $a_{k-1}$ matchings.
\item $p_1$ is matched to $p_{2i}$ with $1 \leq i \leq k$.
This contributes $\sum_{1 \leq i \leq k } c_{i-1} c_{k-i}$ matchings.
\item $p_1$ is matched to $p_{2i+1}$ with $1 \leq i \leq k-1$, $p_{2i}$ and $p_{2i+2}$ are not matched to each other.
This contributes $\sum_{1 \leq i \leq k-1 } a_{i-1} a_{k-i-1}$ matchings.
\item $p_1$ is matched to $p_{2i+1}$ with $1 \leq i \leq k-1$, $p_{2i}$ and $p_{2i+2}$ are matched to each other.
This contributes $\sum_{1 \leq i \leq k-1 } c_{i-1} c_{k-i-1}$ matchings.
\end{enumerate}
This gives
\begin{equation}\label{eq:c}
c_k = a_{k-1} + \sum_{1 \leq i \leq k } c_{i-1} c_{k-i} + \sum_{1 \leq i \leq k-1 } a_{i-1} a_{k-i-1} + \sum_{1 \leq i \leq k-1 } c_{i-1} c_{k-i-1} .
\end{equation}
After simplifying equations \thetag{\ref{eq:a}--\ref{eq:c}}, we obtain:
\begin{gather*}
a_k = c_k - c_{k-1} +
\sum_{i=0 }^{ k-1 } b_{i} c_{k-1-i} +
\sum_{i=0 }^{ k-1 } c_{i} a_{k-1-i} +
2\sum_{i=0 }^{ k-2 } b_{i} c_{k-2-i} +
\sum_{i=0 }^{ k-2 } c_{i} a_{k-2-i} +
\sum_{i=0 }^{ k-3 } b_{i} c_{k-3-i}\\
b_k = c_k +
\sum_{i=0 }^{k-1 } c_{i} b_{k-1-i} +
\sum_{i=0 }^{k-1 } a_{i} c_{k-1-i} +
\sum_{i=0 }^{k-2 } c_{i} b_{k-2-i}
\\c_k = a_{k-1} +
\sum_{i=0}^{ k-1 } c_{i} c_{k-1-i} +
\sum_{i=0}^{ k-2 } a_{i} a_{k-2-i} +
\sum_{i=0}^{ k-2 } c_{i} c_{k-2-i}
\end{gather*}
We translate these equations into generating functions and obtain the
following system, where
we write $A, B, C$ for $A(x), B(x), C(x)$:
\begin{gather*}
A = C ((1-x) + x(1+x)A + x(1+x)^2B )
\\
 B = C ( 1 + xA + x(1+x) B)
\\
C = 1 + xA + x^2A^2 + x(1+x)C^2
\end{gather*}
%
%
%
%
We eliminate $A$ and $B$ from this system and find that $C$ satisfies the equation
  \begin{multline}\label{eq:eqc}
1
-(1+3x+5x^2)C
+x(5+8x+8x^2+9x^3)C^2
-8x^2(1+x)(1+x+x^3)C^3
+{}\\{}
+4x^3(1+x+x^3)(1+x)^2C^4
=0,
  \end{multline}
and that $A$ and $B$ are related to $C$ as follows:
\begin{gather*}
A = \frac{C(1 - x + 2x^2C + 2x^3C)}{1-2xC-2x^2C}
\\
B = \frac{C(1 - 2x^2C )}{1-2xC-2x^2C}
\end{gather*}
Equation~\eqref{eq:eqc} has four solutions.
Only one of them can be written as a formal power series:
\[
\resizebox{\hsize}{!}
{
$C = 
\frac{2(1+x+x^3)-
\sqrt{
2(1+x+x^3)
\left(
1-2x-8x^2-3x^3+
(1+x)\sqrt{(1-x-3x^2)(1-9x-3x^2)}
\right)
}
}
{4x(1+x)(1+x+x^3)}.$
}
\]
The other three solutions have different combinations of signs before
the two square roots. For those combinations, the numerator has a
non-zero
constant term, and this cannot balance the absence of a constant term
in the denominator.
For the series $C(x)$ given above,
the singularity
 closest to $0$ occurs in $\mu=\frac{\sqrt{93}}{6}-\frac{3}{2}$,
one of the roots of $1-9x-3x^2$.
It is a square root singularity,
and there is no other singularity with the same absolute value;
thus,
by the exponential growth formula
~\cite[Thm. IV.7]{flajolet}
and a transfer theorem \cite[Thm. VI.1]{flajolet},
the asymptotics of the sequence is $c_k 
=\Theta((1/\mu)^k \, k^{- 3/2})$
with
$1/\mu =
  (\sqrt{93}+ 9)/2
 \approx 9.3218 $.

Since $c_k$ counts matchings of $2k$ points, it follows that
the number of down-free matchings of $\mathrm{SZZC}_n$
of this kind is $\Theta(\lambda^n / n^{3/2}) $,
where $\lambda = \sqrt{1/ \mu} = \sqrt{
  (\sqrt{93}+ 9)/2
} \approx 3.0532$.
It is easy to see that the same bound holds for all kinds of zigzag chains:
for the proof, note that
a zigzag chain of kind C with $2k$ points
includes
a zigzag chain of kind A with $2k-1$ points
and is included in
a zigzag chain of kind A with $2k+1$ points;
similarly for kind B.

Finally, it follows from Theorem~\ref{thm:double} that
the number of perfect matchings of $\mathrm{DZZC}_n$ (of either kind)
is $\Omega(\lambda^n / n^{4})$
and $O(\lambda^n / n^{3})$.
This proves Theorem~\ref{thm:main}.

\section{$r$-chains without corners}\label{sec:r_chains}
\subsection{Definition of $r$-chains with and without corners}\label{sec:r_chains_def}

In the following two sections we deal with further generalizations of the double chain.
An upward single chain will be called an \emph{arc}.
As usual, the size of an arc is the number of its points.
Recall that three points with  distinct $x$-coordinates
are in \emph{upward position} if they form
a clockwise oriented triangle when sorted by $x$-coordinate.

We define an \emph{$r$-chain (with corners) with $k$ arcs},
to be denoted by $\mathrm{CH}(r,k)$,
see Figure~\ref{fig:ch_with_corners}(a) for an example.
It consists of $k$ arcs of size $r+1$,
the rightmost point of the $i$th arc
($1 \leq i \leq k-1$)
coinciding with
the leftmost point of the $(i+1)$st arc,
so that
any three points are in upward position
if and only if they belong to the same arc.
An $r$-chain $\mathrm{CH}(r,k)$ has $rk+1$ points.
As a special case, a simple (downward) chain is a $1$-chain,
and an even zigzag chain of odd size is a $2$-chain.

One can construct an $r$-chain $\mathrm{CH}(r,k)$ with $k$ arcs
as follows:
\begin{itemize}
\item
Take $k+1$ points $V_{0}, V_{1}, V_{2}, \dots, V_{k}$,
sorted by $x$-coordinate, in downward position.
These points will be called the \emph{corners}.
\item For each $i = 1, 2, \dots, k$,
add $r-1$ points
on the segment $V_{i-1} V_{i}$.
\item Replace each segment $V_{i-1} V_{i}$
by a very flat upward circular arc through $V_{i-1}$ and $V_{i}$.
Move the $r-1$ points from the segment
vertically upwards so that they lie on this circular arc.
The radius of the circular arc must be sufficiently big so that
the orientation of triples of points that do not lie on the same segment
is not changed.
\end{itemize}
We shall often use a compact schematic
drawing of $r$-chains as in Figure~\ref{fig:ch_with_corners}(b).
 In such drawings we have to
draw some matching  edges
as curved lines rather than as straight-line segments.
\begin{figure}[htb]
\centering
\includegraphics[scale=0.9]{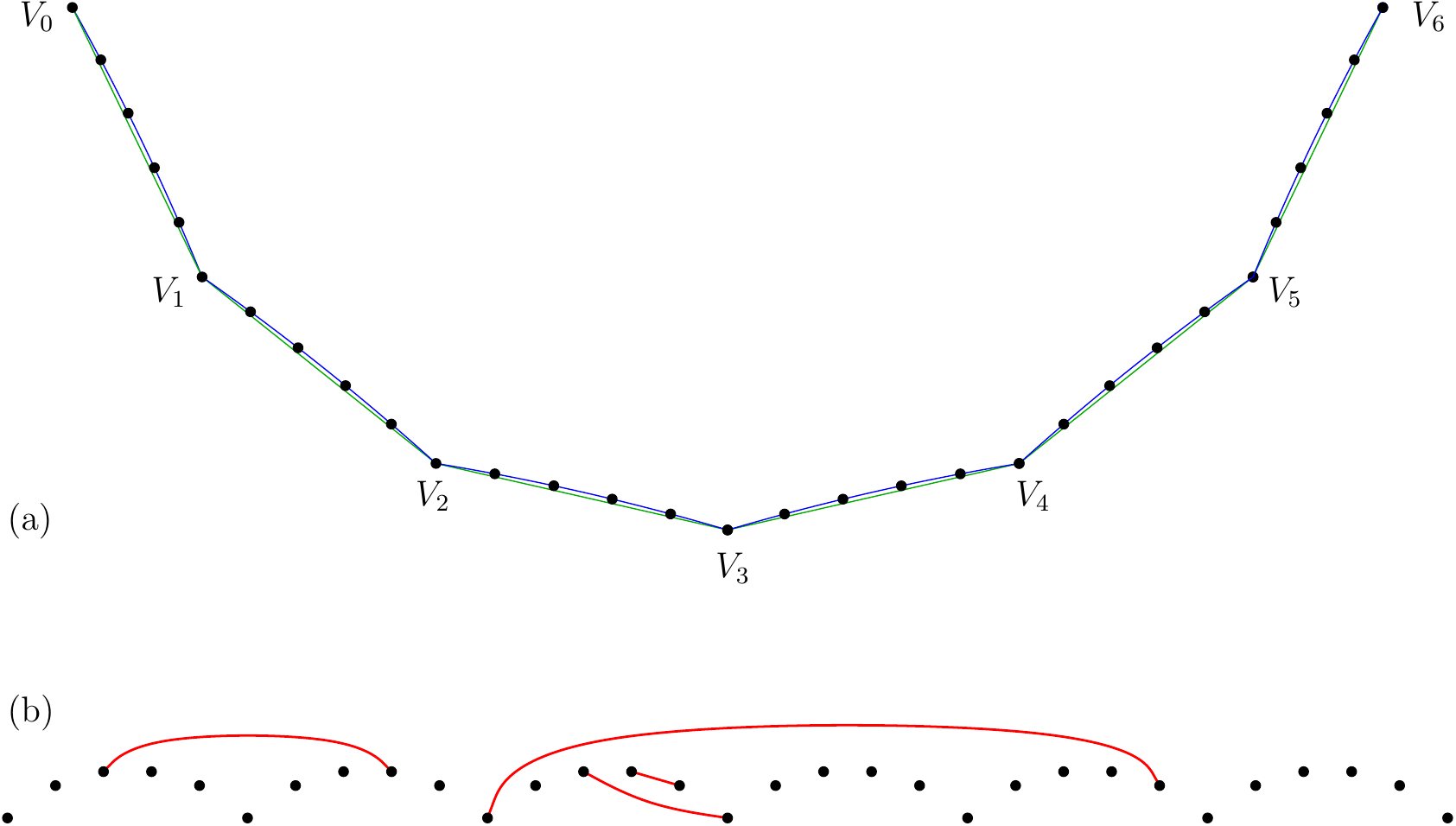}
\caption{A $5$-chain (with corners) with six arcs: (a) a precise drawing; (b) a schematic drawing.}
\label{fig:ch_with_corners}
\end{figure}

The class of
(double)
$r$-chains was 
 earlier used for finding lower bounds
on the maximal number of \emph{triangulations} ($\mathsf{tr}$)
of point sets in the plane.
Garc\'ia, Noy, and Tejel~\cite{garcia} showed that
$\mathsf{tr}(\mathrm{DC}_n)=\Theta^*(8^n)$.
Aichholzer,
Hackl, Huemer, Hurtado, Krasser, and Vogtenhuber~\cite{aich1}
 improved this bound by showing that
$\mathsf{tr}(\mathrm{DZZC}_n)=\Theta^*(8.48^n)$.
This result was further improved by Dumitrescu,
Schulz, Sheffer, and T\'oth~\cite{dum},
who showed that a double $4$-chain of size $n$,
denoted in their work by $D(n, 3^{n/8})$,
has $\Omega(8.65^n)$ triangulations.

Now we define a variation of this structure whose analysis is easier.
An \emph{$r$-chain without corners with $k$ arcs},
denoted by
$\mathrm{CH}^*(r,k)$,
is a set obtained from
$\mathrm{CH}(r+1,k)$
by deleting the corners.
It consists of $rk$ points.
See Figure~\ref{fig:ch_without_corners} for an example.
In this section, we will analyze
$r$-chains without corners, and we will find precise asymptotic
estimates for the number of down-free matchings.
In the next section, we will turn to
$r$-chains with corners. They give even stronger lower bounds, but the analysis will not be so precise.

\begin{figure}[htb]
\centering
\includegraphics[scale=0.9]{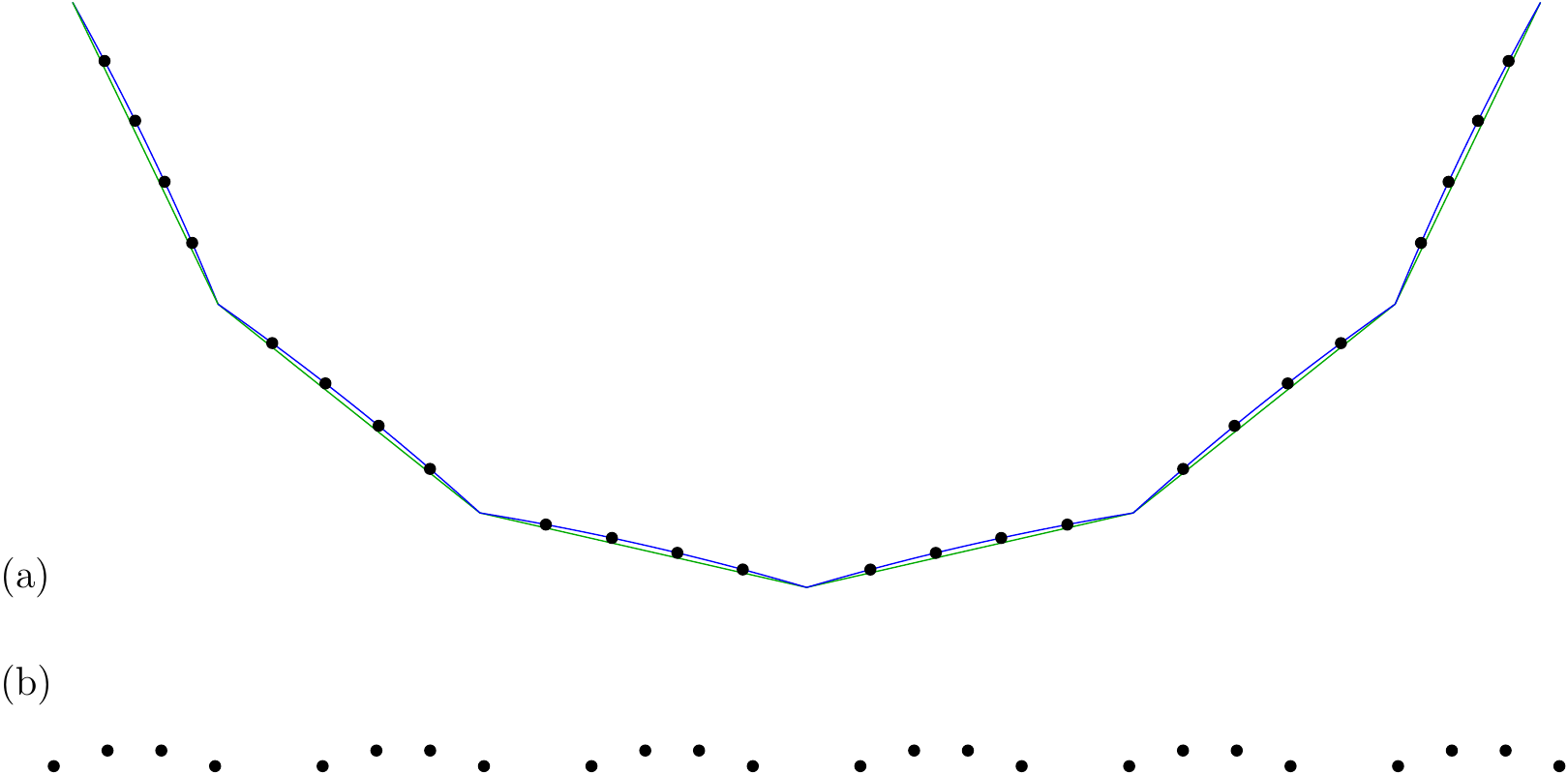}
\caption{A $4$-chain without corners with six arcs: (a) precise drawing; (b) schematic drawing.}
\label{fig:ch_without_corners}
\end{figure}

\subsection{Matchings with runners}
\label{sec:unf_runners}

Consider a matching of $X=\mathrm{CH}^*(r,k)$.
We want to build down-free matchings incrementally from left to right by adding one arc at a time.
If we cut such a matching between two arcs, then we possibly have some edges
cut into two ``half-edges'', which we call \emph{runners}.
(In botany, runners are shoots that connect individual plants.)
A runner
can be formally defined
as a \emph{marked vertex}.
Such a vertex must not be matched by ``proper'' edges and must be visible from above.
 These requirements ensure that, in the course of the incremental construction,
 two runners  can be joined into one edge.
Runners are visualized as half-edges that have one endpoint in $X$
and the other end dangling, see Figure~\ref{fig:runners}(a).
Note that it is not assigned in advance whether a runner will be matched to the left or to the right.

A matching which possibly has runners will be called a
\emph{$\rho$-matching}.
Extending our previous definition of free points, we call
a point
 \emph{free} in a $\rho$-matching
if it is neither matched by a ``proper'' edge
nor marked as an endpoint of a runner.
A $\rho$-matching is \emph{down-free} if all free vertices
are visible from below.

\begin{figure}[htb]
\centering
\includegraphics[scale=1]{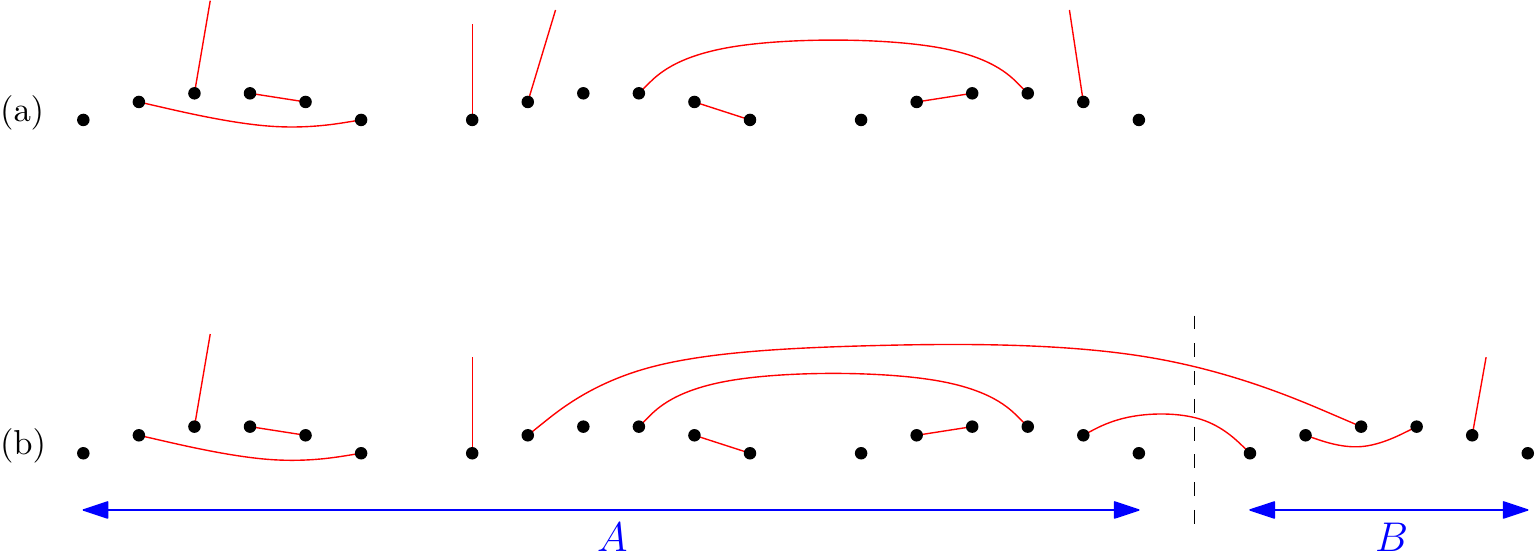}
\caption{(a) A down-free $\rho$-matching $M_A$ with four runners of $A=\mathrm{CH}^*(6,3)$.
(b) Combining  $M_A$ with a down-free $\rho$-matching with three runners of $B=\mathrm{CH}^*(6,1)$.}
\label{fig:runners}
\end{figure}

In the course of the recursive construction
of down-free $\rho$-matchings,
runners from different arcs can be matched
 as follows.
Let $A$ and $B$ be two $r$-chains without corners,
and let $M_A$ and~$M_B$ be down-free $\rho$-matchings of these sets.
We place $B$ to the right of $A$.
If $M_A$ has $j$ runners and $M_B$ has $\beta$ runners,
then for each $\ell$  in the range $0 \leq \ell \leq \min\{j, \beta\}$
we can match, in a unique way,
the rightmost $\ell$ runners of $M_A$ to the leftmost $\ell$ runners of $M_B$.
The obtained $\rho$-matching $M$ is also down-free;
the runners which were not matched in this procedure remain runners in~$M$;
the number of such runners is $j+\beta - 2\ell$.
Conversely,
each down-free $\rho$-matching of $A \cup B$ can be obtained by this procedure
from two uniquely determined down-free $\rho$-matchings of $A$ and $B$.
Figure~\ref{fig:runners}(b) shows an example with
$j=4$, $\beta=3$, $\ell=2$.

We summarize these observations for the special case that we will use in the recursive construction
of $r$-chains: adding one new arc to the right of a given $r$-chain,
 see Figure~\ref{fig:shoots}.

\begin{figure}[htb]
\centering
\includegraphics[scale=1]{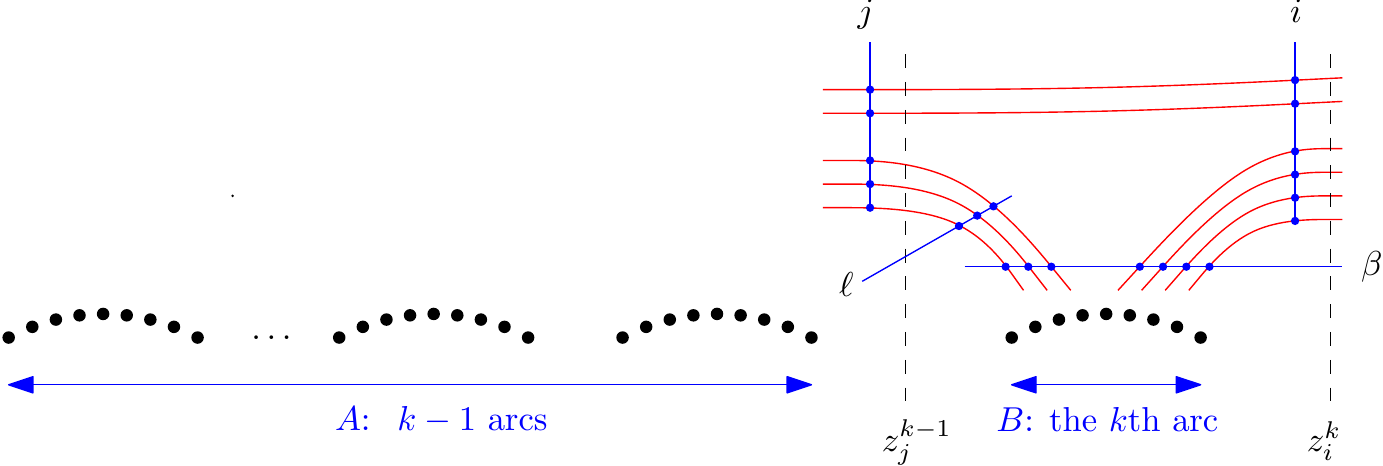}
\caption{Runners in a recursive construction of a $\rho$-matching of an $r$-chain without corners.}
\label{fig:shoots}
\end{figure}

\begin{observation}\label{obs:runners}
Let $X=\mathrm{CH}^*(r,k)$ be an $r$-chain without corners with $k\geq 1$ arcs.
  Let $B$ be the rightmost arc of $X$, and let $A = X \setminus B$.
 Let $M_A$ be a down-free $\rho$-matching of $A$ with $j$ runners,
 and let $M_B$ be a down-free $\rho$-matching of $B$ with $\beta$ runners.
 For each $0 \leq \ell \leq \min\{j, \beta\}$
 there exists a unique down-free $\rho$-matching $M_{X,\ell}$ of $X$
 obtained by matching the rightmost $\ell$ runners of $M$
 with the leftmost $\ell$ runners of~$N$.
 The number of runners in $M_{X,\ell}$ is $i = j+\beta - 2\ell$.

Conversely, each down-free $\rho$-matching $M$ of $X$
can be obtained in this way from uniquely determined matchings $M_A$ and $M_B$ (of $A$ and $B$) as above.
If $M$ has $i$ runners,
$M_A$ has $j$ runners,
and $M_B$ has $\beta$ runners,
then the number of edges obtained by matching of pairs of runners
is $\ell = (j+\beta-i)/2$.
\end{observation}

For $k=1$, this observation holds trivially:
 $A$ is empty, and the only possibility is $j=\ell=0$, $\beta=i$.
 From the above relations between the parameters $i,j,\beta,\ell$,
one can work out the constraints on the possible values of $\beta$
for given $i$ and $j$:
 The equation $i = j+\beta - 2\ell$
together with $0 \leq \ell \leq \min\{j, \beta\}$
implies that $\beta$ must satisfy
$|i-j| \leq \beta \leq i+j$ and $\beta \equiv i-j \pmod{2}$.

\subsection{Recursion for matchings with runners in $r$-chains without corners}
\label{sec:r_chains_wo_analysis}

Denote the number of down-free $\rho$-matchings of $\mathrm{CH}^*(r,k)$
with $i$  runners by $z^k_i(r)$
or simply by $z^k_i$,
since we will regard $r$ as fixed.
Obviously, the down-free matchings of $X=\mathrm{CH}^*(r,k)$
are the down-free $\rho$-matchings without runners.
Since the size of $X$ is $rk$, the growth rate for the number of its down-free matchings
is $\sqrt[r]{z^k_0(r)}$.

For $k=0$ we have $z^0_0=1$ and $z^0_i=0$ for $i>0$.
The numbers $z^1_i$ for a single arc 
will serve as a basis of the recursion. They
are determined in the following proposition.

\begin{proposition} \label{thm:binom_gf}
\begin{enumerate}
  \item
\label{part1}
 The number of down-free matchings (without runners) of a single arc of size $r$ is
  \[ z^1_0= z^1_0(r)=
\binom{r}{\lfloor r/2 \rfloor}.\]
  \item
\label{thm:binom_gf_general}
 The number of down-free $\rho$-matchings of a single arc of size $r$
 that have $i$ runners
is \[ z^1_i= z^1_i(r)=
\binom{r}{i}
\binom{r-i}{\lfloor (r-i)/2 \rfloor}=
\binom{r}{i, \lfloor (r-i)/2 \rfloor, \lceil (r-i)/2 \rceil}.\]
\end{enumerate}
\end{proposition}

\begin{proof}
1. For the first equation, 
let $f(x) = \sum_{r=0}^\infty z^1_0(r)x^r$
 be the generating function for the number of such matchings
in terms of the size $r$ of an arc.
We will show that $f(x)$ satisfies the equation
\begin{equation}
  \label{generating}
f(x)=\frac{1}{1-x}\left(
x^2\cdot c(x^2) \cdot f(x) + 1
 \right),
\end{equation}
where $c(x)=(1-\sqrt{1-4x})/2x$ is
the generating function of the Catalan numbers.
Therefore, we have
 \[f(x) =
 \frac{1}{1-x-x^2c(x^2)},
 \]
 and this is known to be the generating function for $\binom{r}{\lfloor r/2 \rfloor}$ \cite[A001405]{oeis}.

To see why \eqref{generating} holds,
consider the leftmost matched point $p$ (if there is any).
Suppose that $p$ is matched with $q$,
see Figure~\ref{fig:binom_gf} for illustration.
Then all points to the left of $p$ are free,
which contributes $1/(1-x)$ to the generating function.
The points between $p$ and $q$ are
not visible from below
and, therefore, they are
matched by a \emph{perfect} matching; this contributes
$c(x^2)$.
Finally, the points to the right of $q$ are matched by a down-free
matching, whose generating function is {again} $f(x)$.
The factor $x^2$ accounts for the two points $p$ and $q$, and the
additive term $+1$ accounts for the case that $p$ does not exist.

\begin{figure}[htb]
\centering
\includegraphics[scale=0.8]{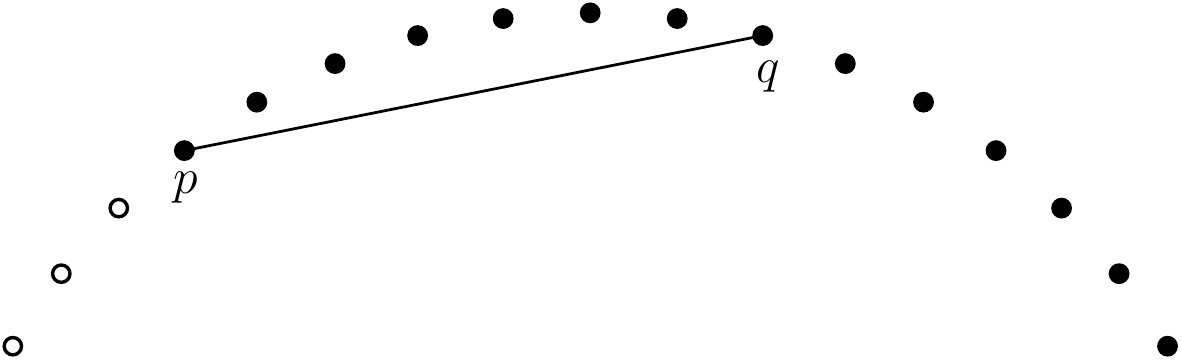}
\caption{The leftmost edge $pq$ in the proof of Proposition~\ref{thm:binom_gf}.\ref{part1}.}
\label{fig:binom_gf}
\end{figure}

We give another proof -- a bijective one.
For a given matching, we mark the left and right endpoints of each
edge by \texttt{L} and \texttt{R}, respectively.  We leave the free
points unmarked for the moment.  Then
the non-crossing matching can be reconstructed from the labels: We
traverse the points from left to right, and we match each \texttt{R}
that we meet with the closest previous unmatched \texttt{L}.
Moreover, since the matching is down-free, free vertices can only
appear when there are no previous unmatched \texttt{L}-vertices.  Now
we label the free points: If there are $\gamma$ free points, we label
the first $\lfloor \gamma \rfloor$ free points by \texttt{R} and the
last $\lceil \gamma \rceil$ free points by \texttt{L}, see
Figure~\ref{fig:01} for illustration.  The free points marked
\texttt{R} can be recovered in a left-to-right sweep as those
\texttt{R}-vertices for which we find no previous matching
\texttt{L}-vertex in the above procedure.  The free points marked
\texttt{L} can be recovered similarly in a right-to-left sweep, and
finally, the matching among the non-free points can be found as
described above.  Thus we have established a bijection with sequences
of length $r$ over the alphabet $\{ \texttt{L}, \texttt{R}\}$ that
contain $\lfloor r/2 \rfloor$ \texttt{R}'s.
\begin{figure}[htb]
\centering
\includegraphics{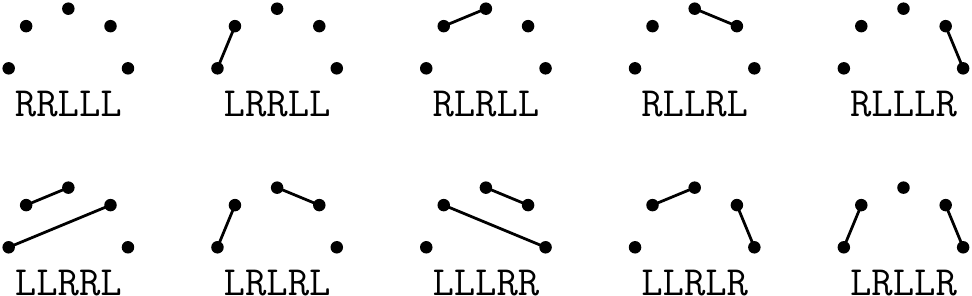}
\caption{The coding of down-free matchings in the second proof of Proposition~\ref{thm:binom_gf}.\ref{part1}.}
\label{fig:01}
\end{figure}

2. Let us turn to the second equation. 
 Once we choose $i$ endpoints of runners,
the whole matching is down-free
if and only if
its restriction on the remaining $r-i$ points is down-free.
Therefore, the result follows directly from the first part. \qedhere
\end{proof}

Now we find a recursion for $z^k_i$, $k \geq 1$.

\begin{proposition}\label{thm:runners_rec}
For fixed $r$, we have the recursion
\begin{equation}\label{eq:rec_wo}
z^k_i
=
\sum_{j \geq 0}
a_{ij}
z^{k-1}_j,
\end{equation}
with coefficients
\begin{equation}\label{eq:sum}
  a_{ij}
=
\sum_{
\substack{
0 \leq \beta \leq r,\\
|i-j| \leq \beta \leq  i+j,\\
\beta  \equiv i-j \pmod{2}
}}
z^1_{\beta}
=
z^1_{ |i-j|}
+
z^1_{ |i-j|+2}
+ \dots +
z^1_{ \min\{r^*, i+j\}},
\end{equation}
where $r^*$ is $r$ or $r-1$ and has the same parity as $i-j$.
\end{proposition}

\begin{proof}
For $k=1$, \eqref{eq:rec_wo} can be verified directly.
Assume now $X=\mathrm{CH}^*(r, k)$ with $k>1$,
let $B$ be the rightmost arc of $X$,
and let $A = X \setminus B$.
For each $j \geq 0$ and each possible $\beta$
we will find the number of $\rho$-matchings of $X$ with $i$ runners
whose restriction to $A$ has $j$ runners
and restriction to $B$ has $\beta$ runners.
By Observation~\ref{obs:runners},
$\rho$-matchings of $A$ and $B$ and the values of
$i$, $j$ and $\beta$ determine uniquely an $\rho$-matching of $X$.
Therefore $\rho$-matchings of $A$ and $B$ with (respectively) $j$ and $\beta$ runners
contribute
$ z^{k-1}_j \cdot z^{1}_\beta$
$\rho$-matchings of $X$ with $i$ runners.

For given $i$ and $j$,
the bounds $|i-j| \leq \beta \leq  i+j$
and the restriction $\beta  \equiv i-j \pmod{2}$
given in~\eqref{eq:sum}
are explained in the remark after Observation~\ref{obs:runners}.
\end{proof}

\subsection{Analysis of the recursion}
\label{sec:recursion-without}

For each $k \geq 0$, denote
$v_k = (z^k_0, z^k_1, z^k_2, z^k_3, \dots)^\top$.
In particular we have $v_0 = (1, 0, 0, 0, \dots)^\top$.
Consider the infinite coefficient matrix $A=(a_{ij})_{i, j \in \mathbb{N}_0}$
with
$a_{ij}$ given by~\eqref{eq:sum}.
By Proposition~\ref{thm:runners_rec},
we have $A v_{k-1} = v_k$ for each $k\geq 1$.
One easily verifies that the matrix $A$ has the following properties:
\begin{itemize}
  \item $A$ is symmetric.
  \item $A$ is a 
band matrix of bandwidth $r$: for 
 $|i-j|>r$ we have $a_{i j}=0$.
  \item
The entries of the first row and column are
 $a_{i0}=a_{0i}=z^1_i = \binom{r}{i}
\binom{r-i}{\lfloor (r-i)/2 \rfloor}$.
  \item For $i+j\geq r^*$ we have  $a_{i+1, j+1} = a_{i j}$.
That is, the diagonals -- sets of entries with fixed $q:=j-i$, $|q| \leq r$ -- \emph{stabilize}
starting from the entry
$a_{(r^*-q)/2,  (r^*+q)/2}$.
    For these entries we have:
\begin{equation}\label{eq:multi}
 a_{ij}=
 a_{i,i+q}=
\sum_{\substack{
|q| \leq \beta \leq r \\
\beta  \equiv q \pmod{2}}
}z^1_\beta.
\end{equation}
    In particular,
starting from the $r$th row (resp.\ column),
  the rows (resp.\ columns) are shifts of each other, and therefore, have the same sum of elements.
  \item The elements in the upper-left corner
  ($i+ j < r^* $) are
positive and
smaller than the elements in the same diagonal after stabilization --
since in this case we have a partial sum of~\eqref{eq:multi}.
\end{itemize}
For example, for $r=5$, the matrix is
\begin{equation}
  \label{eq:Matrix-A}
\setcounter{MaxMatrixCols}{20}
A =
\begin{pmatrix}
10&30&30&20&5&1&0&0&0&0&0&\cdots \\
30&40&50&35&21&5&1&0&0&0&0&\cdots \\
30&50&45&51&35&21&5&1&0&0&0&\cdots \\
20&35&51&45&51&35&21&5&1&0&0&\cdots \\
5&21&35&51&45&51&35&21&5&1&0&\cdots \\
1&5&21&35&51&45&51&35&21&5&1&\cdots \\
0&1&5&21&35&51&45&51&35&21&5&\cdots \\
0&0&1&5&21&35&51&45&51&35&21&\cdots \\
0&0&0&1&5&21&35&51&45&51&35&\cdots \\
0&0&0&0&1&5&21&35&51&45&51&\cdots \\
0&0&0&0&0&1&5&21&35&51&45&\cdots \\
\vdots&\vdots&\vdots&\vdots&
\vdots&\vdots&\vdots&\vdots&
\vdots&\vdots&\vdots&
\ddots
\end{pmatrix}
\end{equation}

The column sum $\lambda_r$ after stabilization of the columns,
that is, starting from the $(r+1)$st column, is as follows:
\begin{equation}
  \label{eq:lambda-r}
\lambda_r =
\sum_{i=0}^r (i+1)z^1_i
 =\sum_{i=0}^r (i+1)  \binom{r}{i}\binom{r-i}{\lfloor (r-i)/2 \rfloor}
=\sum_{i=0}^r (i+1)  \binom{r}{i, \lfloor (r-i)/2 \rfloor, \lceil (r-i)/2 \rceil}
\end{equation}

In the spirit of the Perron-Frobenius theorem for
  non-negative stochastic matrices, one can expect that
$\lambda_r$ is the growth rate for $(z_0^k)_{k \geq 0}$.
We will prove the this is indeed the case
by using a result by Banderier and Flajolet~\cite{ban}
about enumeration of certain kinds of colored lattice paths.

\begin{proposition}\label{thm:rate_wocorners}
For fixed $r$, we have
$z^k_0 = \Theta^*((\lambda_r)^k)$.
\end{proposition}
Note that the superscript $k$ in the left-hand side denotes an index,
whereas in the right-hand side it is a power.

\begin{proof}
We begin with some notion for lattice paths.
Families of lattice paths are usually defined by indicating a starting point -- normally $(0,0)$ -- and a set of possible moves of the form $(1,\beta)$. For many familiar families it is additionally required that the
paths
never go below the $x$-axis and/or end at the $x$-axis.
The paths that
start at $(0,0)$ and
satisfy both these restrictions are called \emph{excursions}.
For example, Motzkin paths~\cite[A001006]{oeis}
are excursions that
use the moves $(1,1)$, $(1,0)$, $(1,-1)$.

In a more general setting,
a set of possible moves may depend on the point reached by a path.
Moreover, 
each move $(1, \beta_i)$ starting in certain point can have
a non-negative integer 
 \emph{multiplicity} $m_i$. 
This is sometimes
expressed by saying that these are copies of the same move
that are distinguished by $m_i$ different ``colors''.

In summary, to each lattice point $(a,b)$ we assign
a \emph{rule} --
a set of moves that can be used for the next step once a path reached this point,
together with multiplicities.
It is assumed that for each lattice point the number of moves with non-zero multiplicity is finite.
Note that one can express the condition of non-crossing the $x$-axis
in terms of such rules:
one has to require that
for each point $(a,b)$ there are only moves $(1, \beta)$ with $\beta \geq -b$.

Consider now the case that
all points that lie on the same horizontal line have the same rule.
Namely,
for $y=j$ and $i\geq 0$
we denote by $d_{ij}$ the multiplicity of the move $(1, i-j)$ at (any)
point $(a, j)$.
We collect these data in the infinite matrix $D=(d_{ij})_{i, j \in
  \mathbb{N}_0}$.
Let $u=(1,0,0,\dots)^\top$.
Then the number of paths that start at $(0,0)$,
do not cross the $x$-axis, and end at a point $(a,b)$
is {equal to} the $b$th component of $D^au$ -- this follows directly from matrix multiplication.
In particular, the upper-left entry of $D^k$ is the number of
 excursions of length $k$, which we will denote by
 $\mathrm{Ex}(D,k)$.
The quantity in which we are interested, the number
$z^k_0$ of down-free matchings, is then given by
$z^k_0=\mathrm{Ex}(A,k)$, where $A$ is the coefficient matrix given above~\eqref{eq:multi}.

Suppose now that we have an even more restricted case:
all points have the same rules;
yet still we want to consider only paths that remain weakly above the $x$-axis,
so we exclude the moves that violate this requirement.
For such families, a result of~\cite[Theorem 3]{ban} can be applied.
It states that
the number of excursions of length $k$
with moves $\{(1, b_1), (1, b_2), \dots, (1, b_m)\}$
and associated multiplicities $w_1,\dots,w_m$,
is of the form $\Theta(C^k/k^{3/2})$,
where the base $C$ of the exponential growth is determined
as follows:
For the Laurent polynomial $P(u) = \sum_ {j=1}^{m} w_j u^{b_j}$, let
 $\tau$ be the unique positive number such that $P'(\tau)=0$;
then $C = P(\tau)$.
The situation is particularly easy for families with a symmetric set of moves,
that is, if $(1, b)$ is a move
then $(1, -b)$ is also a move
with the same multiplicity, or
equivalently, $P(u)=P(u^{-1})$.
In this case, $\tau=1$, and consequently, $C=P(\tau) = \sum_ {j=1}^{m} w_j$.

The situation for our matrix $A$ is very similar to this case, except
 that the first $r-1$ horizontal lines of the lattice
follow different rules, in accordance with the fact that the first $r-1$
rows of $A$ are different from the others.
However, this does not affect the asymptotic growth rate.
Indeed, let us look at the matrix $A'$
in which the first $r$ rows and columns of $A$ have been removed.
It coincides with $A$ for $i+j\geq
r$,
but the rule $a_{i+1, j+1}=a_{i j}$ holds for all entries --
also in the upper-left corner.
Since $A\le A'$ element-wise,
we clearly have
$\mathrm{Ex}(A,k)
\le
\mathrm{Ex}(A',k) = \Theta(\lambda_r^k/k^{3/2})$.
To see that we have a lower bound of the same asymptotic form,
consider only those excursions that start with
the move $(1,+r)$, end with the move $(1,-r)$, and never go below
level $r$. The intermediate part of the excursion is governed by the
matrix $A$ from which the first $r$ rows and columns have been
removed, which coincides with the matrix~$A'$. Thus
 $\mathrm{Ex}(A,k)
\ge
 \mathrm{Ex}(A',k-2) = \Theta(\lambda_r^k/k^{3/2})$.
\end{proof}

\newcommand{\rrr}{}
\newcount\rrr
\rrr=9
\newcommand{\nn}{\global\advance\rrr by 1 {\edef\bbb{\the\rrr}%
\expandafter\nnn \bbb\emd}}
\newcommand{\nnn}[1]{}
\def\nnn#1\emd{%
{\toks0={\def\extract ##1r = #1 |##2 free: ##3 ##4 corner: ##5 ##6 ##7\endxxx
{$\bigl(\begin{smallmatrix}\hfill ##6&\hfill ##5 \\ ##4&
    ##3\end{smallmatrix}\bigr)$}}%
\expandafter\the\toks0
\extract
r = 1 |  F    C, eigenvalue =  3.0         3.00   3.0000
  free: 2 2
corner: 1 1
r = 2 |  F    C, eigenvalue =  3.05316645149         9.32   3.0532
  free: 6 7
corner: 3 3
r = 3 |  F    C, eigenvalue =  3.07109716649        28.97   3.0711
  free: 19 21
corner: 9 10
r = 4 |  F    C, eigenvalue =  3.08187070603        90.21   3.0819
  free: 59 66
corner: 28 31
r = 5 |  F    C, eigenvalue =  3.08767123684       280.64   3.0877
  free: 184 204
corner: 87 97
r = 6 |  F    C, eigenvalue =  3.09088902334       871.97   3.0909
  free: 572 632
corner: 271 301
r = 7 |  F    C, eigenvalue =  3.09245986982      2704.76   3.0925
  free: 1776 1952
corner: 843 933
r = 8 |  F    C, eigenvalue =  3.09300569523      8376.18   3.0930
  free: 5504 6022
corner: 2619 2885
r = 9 |  F    C, eigenvalue =  3.09288169059     25898.21   3.0929
  free: 17030 18550
corner: 8123 8907
r = 10 |  F    C, eigenvalue =  3.09231848233     79954.36   3.0923
  free: 52610 57071
corner: 25153 27457
r = 11 |  F    C, eigenvalue =  3.09146474505    246494.53   3.0915
  free: 162291 175381
corner: 77763 84528
r = 12 |  F    C, eigenvalue =  3.09042030587    758945.49   3.0904
  free: 499963 538386
corner: 240054 259909
r = 13 |  F    C, eigenvalue =  3.08925301802   2333969.82   3.0893
  free: 1538312 1651140
corner: 740017 798295
r = 14 |  F    C, eigenvalue =  3.08800982806   7169707.43   3.0880
  free: 4727764 5059251
corner: 2278329 2449435
r = 15 |  F    C, eigenvalue =  3.08672353634  22002194.62   3.0867
  free: 14514779 15489221
corner: 7006093 7508686
r = 16 |  F    C, eigenvalue =  3.08541724165  67456287.10   3.0854
  free: 44518779 47384904
corner: 21520872 22997907
r = 17 |  F    C, eigenvalue =  3.08410726822 206633661.79   3.0841
  free: 136422462 144857454
corner: 66039651 70382811
r = 18 |  F    C, eigenvalue =  3.08280515072 632454616.15   3.0828
  free: 417702378 442540653
corner: 202462113 215240265
r = 19 |  F    C, eigenvalue =  3.08151899842 1934337016.61   3.0815
  free: 1277945409 1351126551
corner: 620164491 657780918
r = 20 |  F    C, eigenvalue =  3.08025444868 5911965208.22   3.0803
  free: 3907017369 4122747150
corner: 1898109900 2008907469
\endxxx
}}

\subsection{Asymptotic growth constants}
Since $A=\mathrm{CH}^*(r,k)$ has $n=rk$ points,
it follows from Proposition~\ref{thm:rate_wocorners} that
the growth rate for the number of down-free matchings
of the $r$-chain without corners of size $n$
is $\sqrt[r]{\lambda_r}$.
In order to estimate $\lambda_r$, we note that
the expression~\eqref{eq:lambda-r},
when the factor $(i+1)$ is  ignored, counts partitions of $r$ elements
into three subsets (the latter two being of almost equal size).
The total number of such partitions is $3^r$.
Hence, $\lambda_r \leq (r+1)3^r$, and $\sqrt[r]{\lambda_r}$ converges to 3.
Computations suggest that the maximum of $\sqrt[r]{\lambda_r}$ is obtained for $r=11$:
 $\sqrt[11]{\lambda_{11}}=\sqrt[11]{240054} \approx 3.0840$;
 after that it apparently decreases monotonically to $3$, see
 the left part of
 Table~\ref{tab:lambda} for the first few values.
To prove that $r=11$ gives indeed the maximum, one estimates that
 $\sqrt[r]{\lambda_r} \le
3\sqrt[r]{r+1}<3.0838$
for $r\ge191$,
 and the finitely many values up to
 $r=190$ can be checked individually.
%
% ...
This completes the proof of Theorem~\ref{thm:main_wo_corners}.

In order to find a more precise estimate for $\lambda_r$,
we notice that the middle expression
in~\eqref{eq:lambda-r}
expresses $\lambda_r$ as the binomial convolution of the sequence of natural numbers
and the sequence $\binom{m}{\lfloor m/2 \rfloor}
$.
It follows that the exponential generating function for $(\lambda_r)_{r\geq 0}$ is
\[
(1+x) \, e^x \, (I_0(2x)+I_1(2x)),
\]
where $I_0(x)$ and $I_1(x)$
are the modified Bessel functions of the first kind.
From this we can conclude that
the sequence $(\lambda_r)_{r\geq 0}$ is the sum of the sequence A005773 and
a shifted copy of A132894 in~\cite{oeis}.
The ordinary generating function
for this sequence
is then
\[
\frac{1}{2x}
\left(
\frac{1-2x-x^2}
{(1+x)^{1/2}  (1-3x)^{3/2}}
-1 \right),\]
and
it follows from the exponential growth formula
that $\lambda_r = \Theta(3^r r^{1/2})$.
By Theorem~\ref{thm:double}
this number is also the growth rate of
the number of perfect matchings for the corresponding double structure.

\begin{table}[bth]
\renewcommand{\arraystretch}{1.15}
\begin{center}
\begin{tabular}{|r||r|l||c|l|}
\hline
&
\multicolumn2{c||}{without corners}&
\multicolumn2{c|}{with corners}\\
\hline
  $r$& $\lambda_r$& $\sqrt[r]{\lambda_r}$& $M_r$& $T_r$\\
\hline
1&3&3&$\bigl(\begin{smallmatrix}1&1 \\ 2&2\end{smallmatrix}\bigr)$& 3 \\
2&9&3&$\bigl(\begin{smallmatrix}3&3 \\ 7&6\end{smallmatrix}\bigr)$& 3.0532 \\
3&28&3.0366& $\bigl(\begin{smallmatrix}10&\hfill 9 \\ 21&19\end{smallmatrix}\bigr)$& 3.0711 \\
4&87&3.0541& $\bigl(\begin{smallmatrix}31&28 \\ 66&59\end{smallmatrix}\bigr)$& 3.0819 \\
5&271&3.0662&$\bigl(\begin{smallmatrix}\hfill97&\hfill87 \\ 204&184\end{smallmatrix}\bigr)$& 3.0877 \\
6&843&3.0735&$\bigl(\begin{smallmatrix}301&271 \\ 632&572\end{smallmatrix}\bigr)$& 3.0909 \\
7&2619&3.0783&$\bigl(\begin{smallmatrix}\hfill933&\hfill843 \\ 1952&1776\end{smallmatrix}\bigr)$& 3.0925 \\
8&8123&3.0812&
$\bigl(\begin{smallmatrix}2885&2619 \\ 6022&5504\end{smallmatrix}\bigr)$
&\textbf{3.0930} \\
9&25153&3.0829&
$\bigl(\begin{smallmatrix}\hfill 8907&\hfill8123 \\ 18550&17040\end{smallmatrix}\bigr)$& 3.0929 \\
10&77763&3.0837&
\nn
& 3.0923 \\
11&240054&\textbf{3.0840}&
\nn
& 3.0915 \\
12&740017&3.0839&
\nn
& 3.0904 \\
13&2278329&3.0835&
\nn
& 3.0893 \\
14&7006093&3.0829&
\nn
&3.0880 \\
15&21520872&3.0822&
\nn
& 3.0867 \\
16&66039651&3.0813&
\nn
& 3.0854 \\
17&202462113&3.0804&
\nn
& 3.0841 \\
18&620164491&3.0794&
\nn
& 3.0828 \\
19&1898109900&3.0785&
\nn
& 3.0815 \\
20&5805127269&3.0774&
\nn
& 3.0803 \\
 \hline
\end{tabular}
\end{center}
\caption{Summary of results for $r$-chains without and with corners, for $1 \leq r \leq 20$.
For $r$-chains without corners,
$\lambda_r$ is the row sum of the matrix $A$ (Section~\ref{sec:recursion-without}),
and $\sqrt[r]{\lambda_r}$ is the growth rate for $\mathsf{pm}$.
For $r$-chains with corners,
the
\emph{condensed coefficient matrix} $M_r$ is derived from the recursion (Section~\ref{result}),
and $T_r$, the $r$-th root of its dominant eigenvalue,
is the growth rate for $\mathsf{pm}$.
In both cases,
the values for $r=1$ and $r=2$ reproduce the known bounds.
Indeed, a $1$- and a $2$-chain without corners,
as well as
a $1$-chain with corners,
is just a downward chain,
and thus the growth rate of $3$ agrees with Theorem~\ref{thm:gnt}.
A $2$-chain with corners is
a zigzag chain,
and thus $T_2\approx 3.0532$ agrees with Theorem~\ref{thm:main}.
}
\label{tab:lambda}
\end{table}

\section{$r$-chains with corners}\label{sec:with}

\subsection{Recursion}\label{sec:corners_recursion}

In this section, we will treat $r$-chains \emph{with corners},
but we will simply refer to them as $r$-chains.
The analysis of these $r$-chains is more complicated
due to the fact that the corners belong to two arcs.
As before, we will incrementally build the $r$-chain
and estimate the number of matchings of the $r$-chain with $k$ arcs,
which possibly have runners.
We extend the notions of runners, free points, and $\rho$-matchings
to $r$-chains with corners in the obvious way.

We cut a down-free $\rho$-matching $M$ of $\mathrm{CH}(r,k)$
to the \emph{right} of $V_{k-1}$ --
the rightmost point of the $(k-1)$st arc --
and obtain two down-free $\rho$-matchings:
the first, $M_A$, of $A$ --
the set consisting of the first $k-1$ arcs of $\mathrm{CH}(r,k)$;
and the second, $M_B$, of $B$ --
the rightmost arc of $\mathrm{CH}(r,k)$ without the point $V_{k-1}$.
See examples in Figures~\ref{fig:c1}--\ref{fig:f3}.
Note that in the case of $r$-chains with corners
a runner incident to $V_{k-1}$, upon adding $B$ on the right,
can be also connected to a point of $B$:
in such a case we say that it is \textit{matched internally}.

We distinguish whether $M$ has a runner incident to $V_{k}$ or not.
Let $C^k_i$ be the number of down-free $\rho$-matchings of $\mathrm{CH}(r,k)$,
where $V_{k}$ has a runner and \textit{in addition} there are $i$ runners.
Let $F^k_i$ be the number of down-free $\rho$-matchings of $\mathrm{CH}(r,k)$,
where $V_{k}$ has no runner and there are $i$ runners.
($C$ stands for ``corner'', $F$ for ``free''.)
For $k=0$, there is a single vertex, and we have $C^0_0=F^0_0=1$ and
$C^0_i=F^0_i=0$ for all $i>0$.
The number that we are interested in, the number of matchings in
 $\mathrm{CH}(r,k)$, is $F^k_0$. 
Next we find recursive expressions for $C^k_i$ and for $F^k_i$.

\paragraph{Recursion for $C^k_i$.}
For $C^k_i$, the new corner $V_{k}$ has a runner
and is not available for receiving edges from the left.
Thus for the formulae below,
it can be treated as if it were not present
in the $k$th arc.
We have the following three cases:
\begin{enumerate}
\item (Figure~\ref{fig:c1}.) The previous corner $V_{k-1}$ has a runner
which is not matched internally in the $k$th arc.
Suppose there are $\alpha$ runners originating in the $k$-th arc,
in addition to that originating in $V_{k-1}$.
These runners can be only matched to the right.
 The contribution to $C^k_i$ is
\begin{equation}\label{eq:c1}
\sum_{
0 \leq \alpha \leq \min\{r-1,i-1\}
} Z_\alpha
C^{k-1}_{i-1-\alpha},
\end{equation}
where
$$Z_\alpha =
\binom{r-1}\alpha  \binom {r-1-\alpha}{\lfloor(r-1-\alpha)/2\rfloor}.
$$
The expression
for
$Z_\alpha$ is
similar to $z^1_\alpha$
from Proposition~\ref{thm:binom_gf}.\ref{thm:binom_gf_general},
but here we have only $r-1$ points:
all the points of the $k$th arc, excluding the corners.
\begin{figure}[htb]
  \centering
\includegraphics[scale=1]{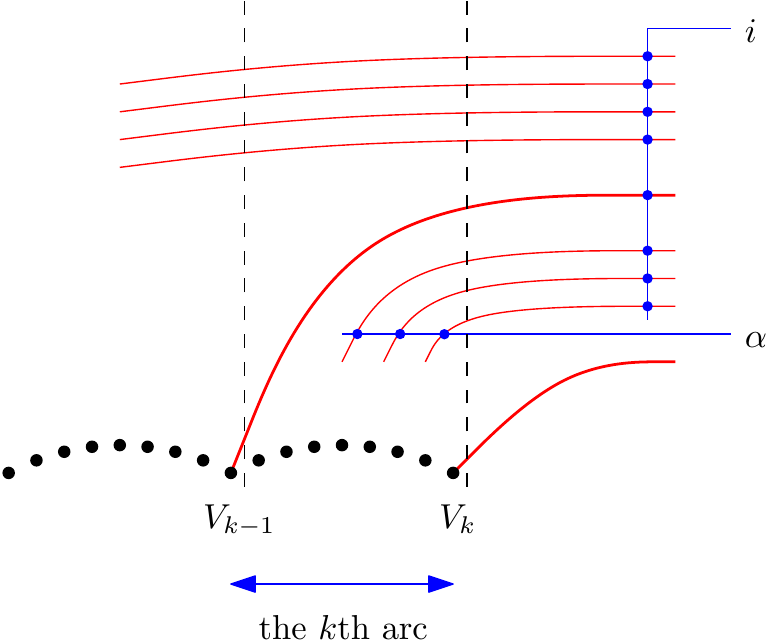}
  \caption{Case 1 in the recursion for $C^k_i$:
  $V_{k-1}$ has a runner not matched internally in the $k$th arc.}
  \label{fig:c1}
\end{figure}

\item (Figure~\ref{fig:c2}.)
$V_{k-1}$ has no runner.
This possibility contributes
\begin{equation}\label{eq:c2}
\sum_{j\ge 0}\, \sum_{
\substack{
|i-j| \leq \alpha \leq  i+j\\
\alpha \equiv i-j \pmod{2}\\
0 \leq \alpha \leq r-1
}} Z_\alpha
F^{k-1}_j.
\end{equation}
This formula (as well as some of the formulae in the following cases)
has the same pattern as \eqref{eq:rec_wo}, with appropriate changes.

\begin{figure}[htb]
  \centering
\includegraphics[scale=1]{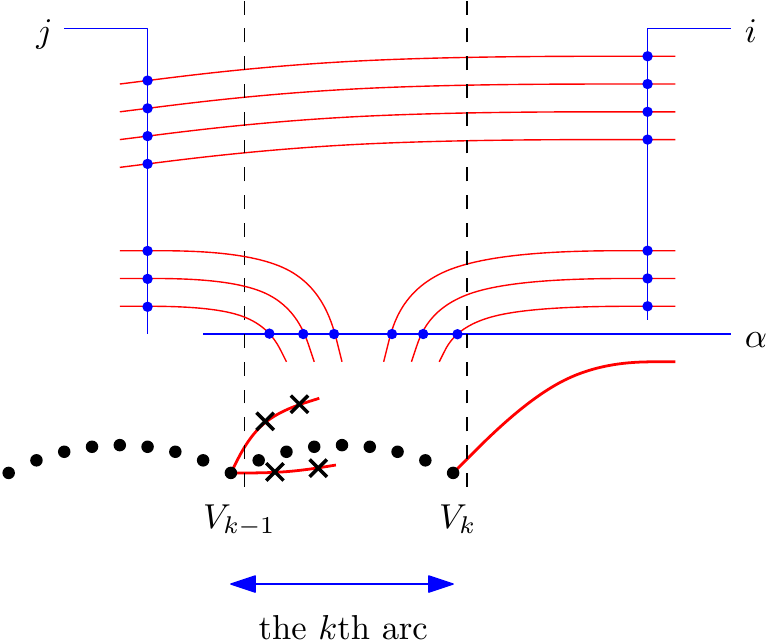}
  \caption{Case 2 in the recursion for $C^k_i$:
  $V_{k-1}$ has no runner.}
  \label{fig:c2}
\end{figure}

\item (Figure~\ref{fig:c3}.) $V_{k-1}$ has a runner matched internally in the $k$-th arc.
The contribution to  $C^k_i$ is
\begin{equation}\label{eq:c3}
\sum_{j\ge 0}\, \sum_{
\substack{
|i-j| \leq \alpha \leq  i+j\\
\alpha \equiv i-j \pmod{2}\\
0 \leq \alpha \leq r-1
}} I_\alpha
C^{k-1}_j,
\end{equation}
where
\[
I_\alpha  =
\binom{r-1}\alpha
         \left[ \binom {r-\alpha}{\lfloor(r-\alpha)/2\rfloor}
         - \binom {r-1-\alpha}{\lfloor(r-1-\alpha)/2\rfloor} \right]
         =
 \binom{r-1}\alpha
         \binom {r-1-\alpha}{\lfloor(r-2-\alpha)/2\rfloor}.
\]
In the expression for $I_{\alpha}$,
the first factor counts the choices of $\alpha$ runners from the $r-1$ points.
In the second factor, we subtract from all down-free $\rho$-matchings on
the remaining
$r-\alpha$ points (including $V_{k-1}$)
those in which $V_{k-1}$ is unmatched, which is the same as
down-free $\rho$ matchings on $r-1-\alpha$ points.

\begin{figure}[htb]
  \centering
\includegraphics[scale=1]{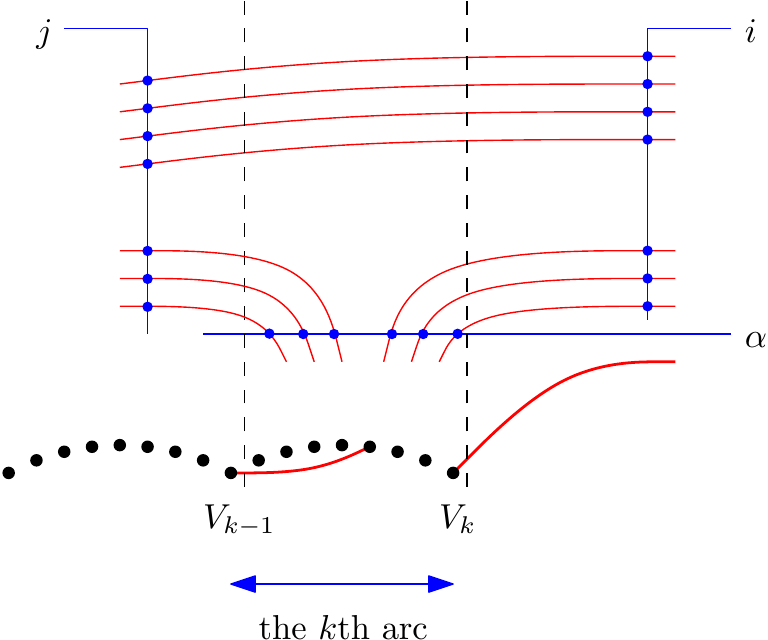}
  \caption{Case 3 in the recursion for $C^k_i$:
$V_{k-1}$ has a runner matched internally in the $k$th arc.}
  \label{fig:c3}
\end{figure}
\end{enumerate}

$C^k_i$ is the sum of the
three expressions \thetag{\ref{eq:c1}--\ref{eq:c3}}.

\paragraph{Recursion for $F^k_i$.}
For $F^k_i$, we have again three cases:
\begin{enumerate}
\item (Figure~\ref{fig:f1}.)
$V_{k-1}$ has a runner not matched internally in the $k$th arc.
In this case, all the additional $\alpha$ runners originating in the
interior of the $k$-th arc must be
matched to the right.
 $V_{k}$ is either free or matched internally to the left.
 The contribution to
 $F^k_i$ is
\begin{equation}\label{eq:f1}
\sum_{
0 \leq \alpha \leq \min\{r-1,i-1\}
} W_\alpha
C^{k-1}_{i-1-\alpha},
\end{equation}
where
$$W_\alpha =
\binom{r-1}\alpha  \binom {r-\alpha}{\lfloor(r-\alpha)/2\rfloor}
$$
is again similar to $z^1_\alpha$ from 
Proposition~\ref{thm:binom_gf}.\ref{thm:binom_gf_general},
but
here we have $r-1$ in the first factor because
no runner originates from $V_{k}$.

\begin{figure}[htb]
  \centering
\includegraphics[scale=1]{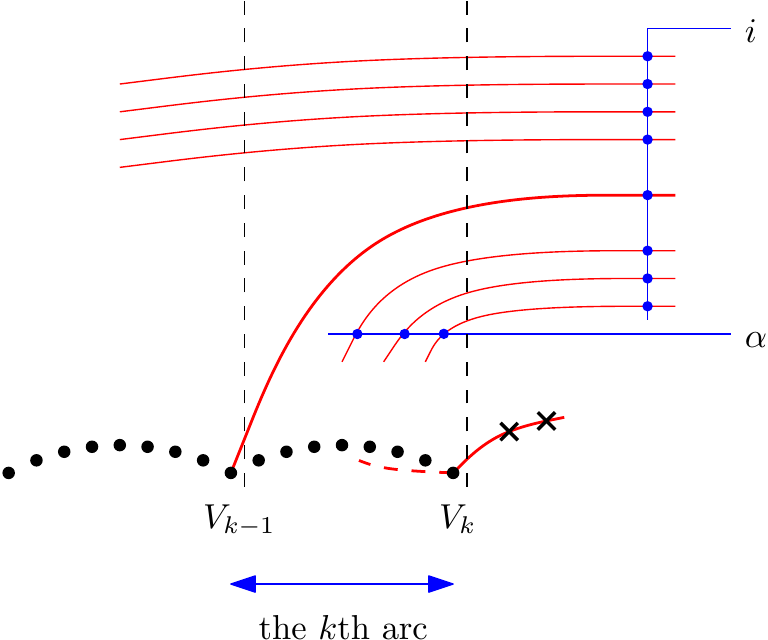}
  \caption{Case 1 in the recursion for $F^k_i$: $V_{k-1}$ has a runner not matched internally in the $k$th arc.}
  \label{fig:f1}
\end{figure}

\item (Figure~\ref{fig:f2}.) $V_{k}$ has a runner 
connected to a point of $A\setminus\{V_{k-1}\}$.
In this case, all $\alpha$ runners originating in the $k$-th arc must be
matched to the left.
The contribution to $F^k_i$ is
\begin{equation}\label{eq:f2}
\sum_{
0 \leq \alpha \leq r-1} (I_\alpha
C^{k-1}_{i+1+\alpha}
+ Z_\alpha
F^{k-1}_{i+1+\alpha}).
\end{equation}

The two terms -- with $C^{k-1}$ and with $F^{k-1}$ --
correspond to the subcases
where $V_{k-1}$ is internally matched or, respectively, not matched
to a point of the $k$th arc.
\begin{figure}[htb]
  \centering
\includegraphics[scale=1]{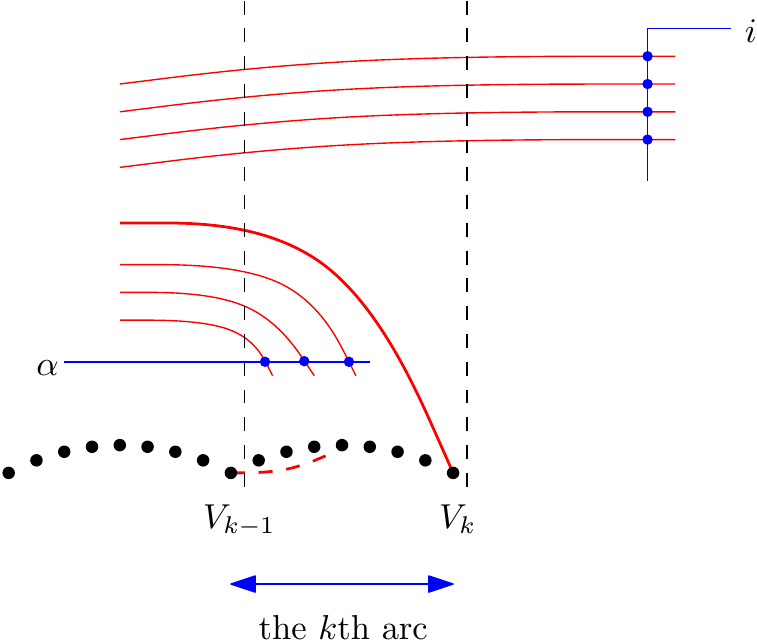}
  \caption{Case 2 in the recursion for $F^k_i$:
  $V_{k}$ 
is connected to a point to the left of $V_{k-1}$.}
  \label{fig:f2}
\end{figure}

\item (Figure~\ref{fig:f3}.) $V_{k-1}$ has no runner, 
and $V_{k}$ 
has no runner matched to a point of $A \setminus \{V_{k-1}\}$.
 The contribution to
 $F^k_i$ is
\begin{equation}\label{eq:f3}
\sum_{j\ge 0} \, \sum_{
\substack{
|i-j| \leq \alpha \leq  i+j\\
\alpha \equiv j-i \pmod{2}\\
0 \leq \alpha \leq r-1
}} (
U_\alpha
C^{k-1}_j+
W_\alpha
F^{k-1}_j),
\end{equation}
where
$$U_\alpha =
\binom{r-1}\alpha
         \left[ \binom {r+1-\alpha}{\lfloor(r+1-\alpha)/2\rfloor}
         - \binom {r-\alpha}{\lfloor(r-\alpha)/2\rfloor} \right]
=
\binom{r-1}\alpha
         \binom {r-\alpha}{\lfloor(r-1-\alpha)/2\rfloor}
$$
The two terms correspond to the same possibilities as in the previous case.
The factor $U_\alpha$
is similar to $I_\alpha$
in the third case for $C_i^k$,
but here we count the down free $\rho$-matchings
of the whole $k$th arc with both its corners,
hence we have $r+1$ instead of $r$ in the second factor.
\begin{figure}[htb]
  \centering
\includegraphics[scale=1]{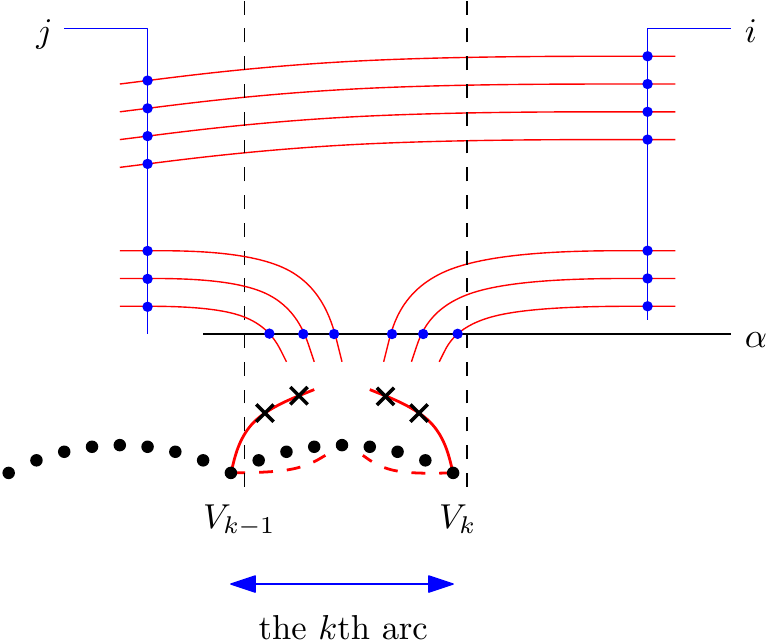}
  \caption{Case 3 in the recursion for $F^k_i$: $V_{k-1}$ has no runner
and $V_{k}$
is not connected to a point to the left of $V_{k-1}$.}
  \label{fig:f3}
\end{figure}
\end{enumerate}
$F^k_i$ is the sum of the
three expressions \thetag{\ref{eq:f1}--\ref{eq:f3}}.

\subsection{Analysis of the recursion}
The expressions above imply 
a coupled mutual recurrence between two sequences of vectors
$C^k = (C^k_0,C^k_1,C^k_2,\ldots)^\top$
and
$F^k =(F^k_0,F^k_1,F^k_2,\ldots)^\top$.
The initial values are
$C^0 = F^0=(1,0,0,\ldots)^\top$.
$C^k$ and $F^k$ are expressed in terms of $C^{k-1}$ and $F^{k-1}$ as
follows.
For $i\ge r$, we have:
\begin{equation}
  \label{eq:coupled}
\begin{aligned}
  C^k_i &=
  \sum_{\beta =-r}^r a^{CC}_{\beta } C^{k-1}_{i+\beta }
+
  \sum_{\beta =-r}^r a^{CF}_{\beta } F^{k-1}_{i+\beta }
\\
  F^k_i &=
  \sum_{\beta =-r}^r a^{FC}_{\beta } C^{k-1}_{i+\beta }
+
  \sum_{\beta =-r}^r a^{FF}_{\beta } F^{k-1}_{i+\beta },
\end{aligned}
\end{equation}
where the numbers $a^{CC}$, $a^{CF}$, $a^{FC}$, $a^{FF}$
are to be read out from the expressions in
Section~\ref{sec:corners_recursion}.
For the small indices $i<r$, we have irregularities, like
for $r$-chains without corners:
The coefficients
in \eqref{eq:coupled}
 must be replaced by smaller coefficients which
depend also on $i$.
In matrix notation, the recursion is written as
\begin{equation}
  \label{eq:coupled_matrices}
\begin{aligned}
  C^k &=
  A^{CC} C^{k-1}
+
  A^{CF} F^{k-1}
\\
  F^k &=
  A^{FC} C^{k-1}
+
  A^{FF} F^{k-1},
\end{aligned}
\end{equation}
where there are four band matrices $A^{CC}$, $A^{CF}$, $A^{FC}$,
$A^{FF}$ 
of bandwidth $r$, similar to the matrix $A$ from~\eqref{eq:Matrix-A}.

\begin{figure}[htb]
  \centering
  \includegraphics{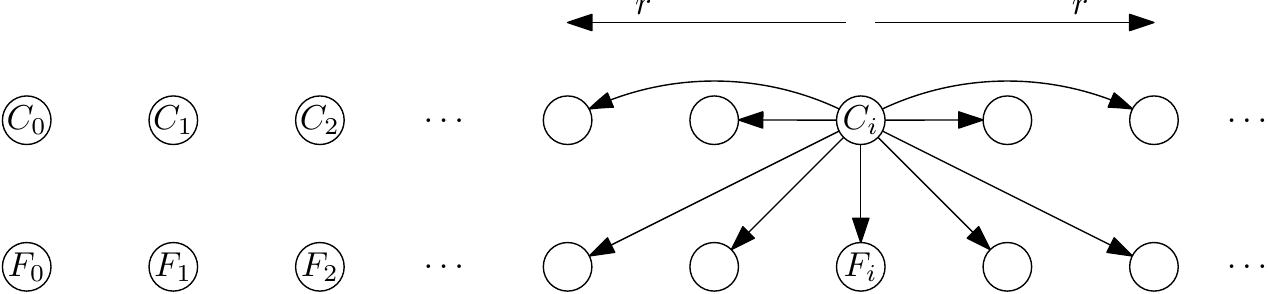}
  \caption{The recursion \eqref{eq:coupled}
 gives the number of paths on this network. The neighborhood of a
 typical vertex $C_i$ is shown in a schematic way.}
  \label{two-layer}
\end{figure}

This system can be interpreted as a set of lattice paths on a
two-layer lattice, see Figure~\ref{two-layer}.
We have a row of nodes
$C_0,C_1,C_2,\ldots$
and another row of nodes
$F_0,F_1,F_2,\ldots$ immediately below it.
The possible jumps and their multiplicity depend only on the row,
with irregularities close to the left edge.
%
(In this representation, the lattice paths considered in the proof of
Proposition~\ref{thm:rate_wocorners} in
Section~\ref{sec:r_chains_wo_analysis} correspond to walks on a ray
$0,1,2,\ldots$.  The $x$-coordinate of the two-dimensional lattice in
Section~\ref{sec:r_chains_wo_analysis} is now represented as time.

We are not able to provide as precise estimates for the growth
rate as for chains without corners, where we had a single recursion.
One would expect a similar behaviour. However, we can still pin
down the base of the exponential growth as an eigenvalue of an
associated $2\times2$ matrix.

First, we can get rid of the irregularities but cutting of the first
$r$ 
 rows and columns of the coefficient matrices. As for the
case of a single matrix, this does not affect the asymptotic growth.
We can now assume that the diagonals are constant, and
the recursion
 \eqref{eq:coupled} holds for all $i$, with the convention that
$ C^{k-1}_j$ and
$ F^{k-1}_j$ in the right-hand side
are taken as $0$ for $j<0$.

For better readability,
we will now replace the vectors $C^k$ and $F^k$ by more generic names
$x^k$ and $y^k$:
\begin{equation}
  \label{eq:coupled-xy}
\begin{aligned}
  x^k_i &=
  \sum_{\beta =-r}^r a^{XX}_{\beta } x^{k-1}_{i+\beta }
+
  \sum_{\beta =-r}^r a^{XY}_{\beta } y^{k-1}_{i+\beta }
\\
  y^k_i &=
  \sum_{\beta =-r}^r a^{YX}_{\beta } x^{k-1}_{i+\beta }
+
  \sum_{\beta =-r}^r a^{YY}_{\beta } y^{k-1}_{i+\beta },
\end{aligned}
\end{equation}
for all $i$, with the understanding that quantities 
$x^{k-1}_j$
and $y^{k-1}_j$
 with negative
subscripts $j$ on the right-hand side are regarded as zero.

Our analysis below is valid under the condition that the coefficients of
this recursion don't exhibit a tendency to favor larger or smaller
indices, or in other words, that the Markov chain associated to the
system does not systematically drift to the left or to the right.
(In the one-vector recursion analyzed in the proof of
Proposition~\ref{thm:rate_wocorners}, this no-drift condition was not
an issue because the set of moves was symmetric.)
To formulate this 
 condition precisely, we have to set up some
notation and establish
some terms.

Let us denote the coefficient sums in the terms of the recursion~\eqref{eq:coupled-xy} as follows:
$$
\bar A^{XX} =
  \sum_{\beta =-r}^r a^{XX}_{\beta }
,\ \bar A^{XY} =
  \sum_{\beta =-r}^r a^{XY}_{\beta }
,\ \bar A^{YX} =
  \sum_{\beta =-r}^r a^{YX}_{\beta }
,\ \bar A^{YY} =
  \sum_{\beta =-r}^r a^{YY}_{\beta }.
$$
These numbers are the 
 column sums of the coefficient matrices after stabilization.
 These sums form the \emph{condensed coefficient matrix}
\begin{equation}
  \label{transition}
  \begin{pmatrix}
\bar A^{XX}
&
 \bar A^{XY} \\
 \bar A^{YX} &
 \bar A^{YY}
  \end{pmatrix}.
\end{equation}
Let $M$ denote its dominant eigenvalue.
Let $(\rho_X,\rho_Y)$ be the corresponding left eigenvector 
and $(\pi_X,\pi_Y)^\top$ be the corresponding right eigenvector, 
with the normalization $\rho_X+\rho_Y=\pi_X+\pi_Y=1$.
Since the matrix is positive, these two vectors are positive.


We define the
\emph{total group-to-group jump sizes} of the system:
\begin{equation*}
D^{XX} =
  \sum_{\beta =-r}^r a^{XX}_{\beta } \beta
,\
D^{XY} =
  \sum_{\beta =-r}^r a^{XY}_{\beta } \beta
,\
D^{YX} =
  \sum_{\beta =-r}^r a^{YX}_{\beta } \beta
,\
D^{YY} =
  \sum_{\beta =-r}^r a^{YY}_{\beta } \beta.
\end{equation*}
The
\emph{weighted total jump size} $D$ of the system is then defined
as follows:
\begin{align}
  \label{total-jump}
D &
=
\rho_X
\pi_X
D^{XX}
+
\rho_X
\pi_Y
D^{XY}
+\rho_Y
\pi_X
D^{YX}
+\rho_Y
\pi_Y
D^{YY}
\\\nonumber
&
=
\begin{pmatrix}
  \rho_X & \rho_Y
\end{pmatrix}
  \begin{pmatrix}
 D^{XX}
&
 D^{XY} \\
 D^{YX} &
 D^{YY}
  \end{pmatrix}
\begin{pmatrix}
  \pi_X \\\pi_Y
\end{pmatrix}
\end{align}

Now we can state the main result of the analysis.
\begin{theorem} \label{arbitrary-row-sums}
Suppose the system
\eqref{eq:coupled-xy}
has non-negative coefficients,
and the {weighted total jump size} $D$
 is zero. 
Assume that the coefficients
$a^{XX}_\beta,a^{XY}_\beta,a^{YX}_\beta,a^{YY}_\beta$ are positive for
$\beta=-1,0,1$. 
Let $M$ be the dominant eigenvalue of the
{condensed coefficient matrix}~\eqref{transition}.
Then
$$x^k_0 = O(M^k),\
y^k_0 = O(M^k),$$
and
$$x^k_0 = \Omega((M-\eps)^k),\
y^k_0 = \Omega((M-\eps)^k)
$$
for every $\eps>0$.
\end{theorem}
Since the proof is quite substantial, we devote a separate section to it.

\subsection{Proof of the theorem
about mutually coupled recursions}

  We will transform the problem to a recursion in which the left
  eigenvector is $(\rho_X,\rho_Y)=
(1,1)$, 
and thus
the column sums of the coefficient matrix (after stabilization) are constant.
We achieve this by rescaling the vectors $x$ and $y$ to
 $\tilde
x^k_i=\rho_X
x^k_i$
and $\tilde
y^k_i=\rho_Y
y^k_i$.
Clearly, the asymptotic growth of $x$ and $y$ is unaffected by this
multiplication with a constant.
For these new vectors, the coefficients of the recursion
change to
$\tilde a^{XY}_{\beta } = \rho_X/\rho_Y \cdot
a^{XY}_{\beta }$
and
$\tilde a^{YX}_{\beta } = \rho_Y/\rho_X \cdot
a^{YX}_{\beta }$,
while
$\tilde a^{XX}_{\beta } = a^{XX}_{\beta }$,
$\tilde a^{YY}_{\beta } = a^{YY}_{\beta }$ are unchanged.
Consequently, the first column sum of the
{condensed coefficient matrix} \eqref{transition} becomes
$
\bar A^{XX}
+\rho_Y/\rho_X \cdot
 \bar A^{XY}
=
(\rho_X \cdot
\bar A^{XX}
+\rho_Y \cdot
 \bar A^{XY})
/\rho_X
=
(M\rho_X)
/\rho_X=
M$, and similarly for the second column.
Theorem~\ref{arbitrary-row-sums} 
follows therefore from the following theorem, 
which is a special case of Theorem~\ref{arbitrary-row-sums} with the
additional assumption that the
matrix \eqref{transition} has constant column sums.

\begin{theorem} \label{equal-row-sums}
Suppose the system
\eqref{eq:coupled-xy}
has non-negative coefficients
and constant column sums
\begin{equation}
  \label{row-sums}
M=\bar A^{XX}
+
 \bar A^{YX} =
 \bar A^{XY}
+
 \bar A^{YY}.
\end{equation}
%
Suppose that
\begin{equation}
  \label{total-jump-tilde}
\pi_X
 \left(
D^{XX}
+D^{YX}
\right)
+
\pi_Y \left(
D^{XY}
+D^{YY}
\right)
= 0,
\end{equation}
%
where $(\pi_X,\pi_Y)$ is a right eigenvector of
the matrix~\eqref{transition}
with eigenvalue $M$.
Suppose further that the coefficients
$a^{XX}_\beta,a^{XY}_\beta,a^{YX}_\beta,a^{YY}_\beta$ are positive for
$\beta=-1,0,1$. 
Then
$$x^k_0 = O(M^k),\
y^k_0 = O(M^k),$$
and
$$x^k_0 = \Omega((M-\eps)^k),\
y^k_0 = \Omega((M-\eps)^k)
$$
for every $\eps>0$.
\end{theorem}
Theorem~\ref{equal-row-sums}
is formulated in terms of the original recursion
\eqref{eq:coupled-xy}, but it must be applied to
 $\tilde x$
and $\tilde y$ instead of $x$ and $y$ in order to prove
 Theorem~\ref{arbitrary-row-sums}.
Our rescaling modifies group-to-group jump sizes in the same way as the
coefficients:
$\tilde D^{XY} = \rho_X/\rho_Y \cdot
D^{XY}$, etc.; the eigenvectors of the modified condensed coefficient
matrix
are
 $(\tilde \pi_X,\tilde \pi_Y)=(\rho_X\pi_X,\rho_Y\pi_Y)$ and
 $(\tilde \rho_X,\tilde \rho_Y)=(1,1)$ (without normalization), and
 with these substitutions, the condition that $D$ from
\eqref{total-jump} is zero translates into
\eqref{total-jump-tilde}, after erasing the tildes.
This concludes the proof of Theorem~\ref{arbitrary-row-sums}.
\qed

\begin{proof}[Proof of Theorem~\ref{equal-row-sums}]
The upper bound is easy: by summing all equations of
\eqref{eq:coupled-xy}, one sees
that $\sum_{i\ge0} x_i^k + \sum_{i\ge0} y_i^k$ can grow at most by the factor $M$ in
each iteration, since the column sums of the coefficient matrix are
bounded by $M$. It follows that
$x_0^k,y_0^k\le\sum_{i
} x_i^k + \sum_{i
} y_i^k
\le M^k(\sum_{i
} x_i^0 + \sum_{i
} y_i^0)
=2 M^k$.

Let us now turn to the lower bound:
To have a compact notation for the linear operator expressing in the
recursion~\eqref{eq:coupled-xy}, we denote it by $\phi$:
\[(x^k,y^k)= \phi(x^{k-1},y^{k-1})\]
As an intermediate lemma, we will show that any ``sub-eigenvector'' with
eigenvalue $\lambda$ is enough for a lower bound on the growth.
\begin{lemma}\label{sub-eigenvector}
  Suppose there is a pair of non-negative non-zero vectors $\bar x$ and
  $\bar y$ with finitely many non-zero elements
such that the inequality
\begin{equation}
  \label{super}
 \phi(\bar x,\bar y) \ge \lambda \cdot(\bar x,\bar y)
\end{equation}
holds componentwise for some $\lambda>0$.
Then there is a constant $K>0$ such that
$x^n_0,y^n_0 \ge K\lambda^n$ for all $n \in \mathbb{N}$.
\end{lemma}
\begin{proof}
Since $\phi$ is a monotone operator,
the inequality \eqref{super} remains fulfilled if we apply $\phi$ on the vectors arbitrary many times:
\begin{equation*}
 \phi^{k+1}(\bar x,\bar y) \ge \lambda \cdot\phi^k(\bar x,\bar y)
\end{equation*}
Applying $\phi$ on $(\bar x,\bar y)$ sufficiently many times,
we eventually obtain a vector whose 
components $\bar x_0$ and $\bar y_0$
 are positive, since 
 the coefficients
$a^{XX}_1,a^{XY}_1,a^{YX}_1,a^{YY}_1$ are positive by assumption.
%
Moreover, by scaling we can obtain a vector in which these components are bigger than $1$.
Thus, we can assume that 
 $\bar x_0\ge1$ and $\bar y_0\ge1$.

Now, we find $n_1$ and $K$ such that the following inequality holds componentwise
for $n=n_1$:
\begin{equation}
  \label{induction}
(x^n,y^n) \ge K\lambda^n\cdot (\bar x,\bar y)
\end{equation}
To see that this is possible, we use the assumption that
$a^{XX}_\beta,a^{XY}_\beta,a^{YX}_\beta,a^{YY}_\beta$ are positive for
$\beta=0$ and $\beta=-1$.
Thus,
by making $n_1$ big enough, we can ensure that
$(x^{n_1},y^{n_1})$ has positive components wherever $(\bar x,\bar y)$ has
 positive components. We can then fulfill
\eqref{induction} by choosing $K$ small enough.

The inequality \eqref{induction} carries over to all larger
$n$ by induction, using
monotonicity of the operator $\phi$ and the assumption~\eqref{super}.
Since $\bar x_0\ge1$ and $\bar y_0\ge1$,
 the desired inequalities
follow from \eqref{induction}
 for all $n\ge n_1$.
Finally, for the finitely many values $n<n_1$, we can fulfill the inequalities
$x^n_0,y^n_0 \ge K\lambda^n$ by decreasing $K$ if necessary.
\end{proof}

Let us explain the idea for getting ``sub-eigenvectors'' $\bar x$ and
$\bar y$ for Lemma~\ref{sub-eigenvector}.  If we wish to fulfill
\eqref{super} for $\lambda=M$, vectors $\bar x$ and $\bar y$ with
constant entries will do the job. However, they have infinitely many
non-zero entries. Thus, we aim for a smaller $\lambda=M-\eps$, and we
make an \emph{ansatz} where the entries are determined by a concave
quadratic function. This has to be adjusted later because the vectors
have to be non-negative, and because the recursion
\eqref{eq:coupled-xy} has some irregularities for the small values
$i<r$.  Moreover, the two coupled sequences $\bar x$ and $\bar y$
depend on each other in a non-symmetric way, and therefore we cannot
use the same quadratic function for both sequences. They have to be
scaled differently, and shifted horizontally relative to each
other. 

We define the shift constant
\begin{equation*}
  \delta =
\frac
{\pi_X D^{XX}+\pi_Y D^{XY}}
{-\pi_Y A^{XY}}
=
\frac
{-(\pi_X D^{YX}+\pi_Y D^{YY})}
{\pi_Y (A^{YY}-M)}
.
\end{equation*}
In this definition, equality of the numerators follows from
the assumption
\eqref{total-jump-tilde}, which expresses
 that
the {weighted total jump size} 
is zero. 
The denominators are equal because the column sums are $M$~\eqref{row-sums}.


We take some real parameters $p$ and $s$, to be determined later, and
define 
the quadratic functions $h_X$ and $h_Y$ 
and two auxiliary vectors  $\hat x$ and $\hat y$ as follows:
  \begin{align}
\label{functions-h}
h_X(i) &= \pi_X(p-i^2) \\
\label{functions-h2}
h_Y(i) &= \pi_Y(p-(i+\delta)^2)\\
\nonumber
\hat x_i&=h_X(i-s)\\
\nonumber
\hat y_i&=h_Y(i-s)
  \end{align}
for all $i\in \mathbb{Z}$. The two quadratic functions have their
peaks at $i=0$ and $i=-\delta$, with respective values $p\pi_X$ and
$p\pi_Y$. These function are shifted to the right by $s$
before they are used as entries of $\hat x$ and $\hat y$.  The
following lemma expresses the difference of these vectors from being
an eigenvector with eigenvalue $M$.

\begin{lemma}\label{constant}
Each of the two expressions
\begin{align}
\label{c}
Q_X&=M \cdot\hat  x_i -\left(
  \sum_{\beta =-r}^r a^{XX}_{\beta } \hat  x_{i+\beta }
+
  \sum_{\beta =-r}^r a^{XY}_{\beta } \hat  y_{i+\beta }
\right),
\\
\label{f}
Q_Y  &=M \cdot\hat  y_i -\left(
  \sum_{\beta =-r}^r a^{YX}_{\beta } \hat  x_{i+\beta }
+
  \sum_{\beta =-r}^r a^{YY}_{\beta } \hat  y_{i+\beta }
\right)
\end{align}
has a constant value, independent of $i$, $p$, and $s$.
\end{lemma}

\begin{proof}
First, we replace the quadratic function in each of the summation
terms
by a Taylor series around the weighted
average point.
The linear terms will then cancel, and the quadratic terms have a
constant value.
 We carry this out by way of example for the
sum of the $a^{XX}$ terms. The parameters $i$ and $s$ always occur
together in the combined term $i-s$, and thus we express our terms
in terms of the parameter $t := i-s$.
$$
  \sum_{\beta =-r}^r a^{XX}_{\beta } \hat  x_{i+\beta }
=
  \sum_{\beta =-r}^r a^{XX}_{\beta } h_X(i+\beta-s)
=
  \sum_{\beta =-r}^r a^{XX}_{\beta } h_X(t+\beta)
$$
Let
\begin{equation*}\nonumber
\bar  D^{XX} =\frac
 { \sum\limits_{\beta =-r}^r a^{XX}_{\beta } \beta}
 { \sum\limits_{\beta =-r}^r a^{XX}_{\beta }}
=\frac
{ D^{XX}}
{ A^{XX}}
\end{equation*}
denote the \emph{average jump size} from group $X$ to group $X$.
Then we rewrite $h_X$ as a Taylor series in the point 
$t+\bar  D^{XX}$.
\begin{equation*}
  h_X(t+x) = h_X(t+\bar  D^{XX})+
h'_X(t+\bar  D^{XX})(x-\bar  D^{XX}) -\pi_X(x-\bar  D^{XX})^2
\end{equation*}
We get
\begin{align*}
 & \sum_{\beta =-r}^r a^{XX}_\beta h_X(t+\beta)   \\ 
& =
h_X(t+\bar  D^{XX})\sum_{\beta} a^{XX}_\beta 
 + 
h'_X(t+\bar  D^{XX})
\sum_\beta a^{XX}_\beta(\beta-\bar  D^{XX})
 -\pi_X\sum_\beta a^{XX}_\beta(\beta-\bar  D^{XX})^2
 \\
& =
h_X(t+\bar  D^{XX})A^{XX} + h'_X(t+\bar  D^{XX})
\cdot0 - \mathrm{const.}
\end{align*}
We transform the other sum in
the expression \eqref{c} analogously, using
 the {average jump size} $\bar  D^{XY}$,
 and we rewrite~\eqref{c} as follows:
\begin{align*}
Q_X &=
M\cdot h_X(t)-
h_X(t+\bar  D^{XX})A^{XX}-
h_Y(t+\bar  D^{XY})A^{XY}
+ \mathrm{const}\\
&=
M\cdot \pi_X(p-t^2)-
\pi_X\bigl(p-(t+\bar  D^{XX})^2\bigr)A^{XX}-
\pi_Y\bigl(p-(t+\bar  D^{XY}+\delta)^2\bigr)A^{XY}
+ \mathrm{const}\\
&=
(p-t^2)(\pi_XM-
\pi_XA^{XX}
-\pi_YA^{XY})
\\&\qquad
+2t(\pi_X \bar  D^{XX}A^{XX}+
\pi_Y \bar  D^{XY}A^{XY}
+\pi_Y \delta A^{XY})
+ \mathrm{const}
\end{align*}
The coefficient of $(p-t^2)$ is zero because $(\pi_X,\pi_Y)$ is an eigenvector,
and the coefficient of $t$ is zero by the definition of $\delta$.
Thus, the expression $Q_X$ has a constant value, as claimed.

For the expression \eqref{f}, the calculation is slightly different.
\begin{align*}
Q_Y &=
M\cdot h_Y(t)-
h_X(t+\bar  D^{XY})A^{XY}-
h_Y(t+\bar  D^{YY})A^{YY}
+ \mathrm{const}\\
&=
M\cdot \pi_Y\bigl(p-(t+\delta)^2\bigr)-
\pi_X\bigl(p-(t+\bar  D^{XY})^2)\bigr)A^{XY}-
\pi_Y\bigl(p-(t+\bar  D^{YY}+\delta)^2\bigr)A^{YY}
+ \mathrm{const}\\
&=
(p-t^2)(\pi_YM-
\pi_XA^{XX}
-\pi_YA^{XY})
\\
&\qquad
+2t(-\pi_Y \delta M+
\pi_X \bar  D^{XY}A^{XY}+
\pi_Y \bar  D^{YY}A^{YY}
+\pi_Y \delta A^{YY})
+ \mathrm{const}
\end{align*}
The coefficients of $(p-t^2)$ and $t$ vanish for the same reasons as
above. This concludes the proof of the lemma.
\end{proof}







The vectors $\hat x$ and $\hat y$ have negative values. 
To get our desired vectors $\bar x$ and $\bar y$, we will clip these
values to~0.
We determine the parameters $p$ and $s$ in such a way that
the resulting vectors $\bar x$ and $\bar y$ start with a big jump from
0 to a positive value, big enough to accommodate the
``perturbation'' resulting from modifying the negative values to 0.
Let $\eps>0$ be given, and
let $K := \max\{Q_X,Q_Y\}$ be the maximum of $Q_X$ and $Q_Y$.
($K$ is positive, but the ensuing argument does not depend on this fact.)
We look at the sorted set of values
$$
\{\,i^2\mid i\in \mathbb{Z}\,\}
\cup
\{\,(i-\delta)^2\mid i\in \mathbb{Z}\,\}
$$
and find $p$ as a positive value in this set such that the gap to the largest value which is
smaller
than $p$ is at least $K/(\eps\min\{\pi_1,\pi_2\})$.
Since the functions are quadratic, there must be larger and larger
gaps as the numbers get bigger, and therefore such a value $p$ exists.
For the functions $h_X
(i) = \pi_X(p-i^2)$ and $h_Y(i) = \pi_Y(p-(i+\delta)^2)$
in~\thetag{\ref{functions-h}--\ref{functions-h2}}, this implies
that the smallest positive value in their range is at least
 $K/\eps$.
Now we shift the functions horizontally
such that positive values occur only at positive arguments,
by choosing $s\ge \sqrt p\pm\delta$.
Finally, we clip the negative values and define
$$
\bar x_i=\max\{h_X(i-s),0\},\
\bar y_i=\max\{h_Y(i-s),0\}
,
$$
for all $i\in \mathbb{Z}$.  This will set $\bar x_i= \bar y_i=0$ for
$i<0$, in accordance with the interpretation that is given in
\eqref{eq:coupled-xy} when these values appear on the right-hand side.

We will show that
\begin{equation}
  \label{bigger}
 \phi(\bar x,\bar y) \ge (M-\eps) \cdot(\bar x,\bar y),
\end{equation}
thus establishing condition \eqref{super} and proving the lower bound
of the theorem with the help of Lemma~\ref{sub-eigenvector}.

In concrete terms, our desired relation
\eqref{bigger} looks as follows:
\begin{align}
  \label{concrete-C}
  (M-\eps) \cdot\bar x_i &\le
  \sum_{\beta =-r}^r a^{XX}_{\beta } \bar x_{i+\beta }
+
  \sum_{\beta =-r}^r a^{XY}_{\beta } \bar y_{i+\beta }
\\
  \label{concrete-F}
  (M-\eps) \cdot\bar y_i &\le
  \sum_{\beta =-r}^r a^{YX}_{\beta } \bar x_{i+\beta }
+
  \sum_{\beta =-r}^r a^{YY}_{\beta } \bar y_{i+\beta }
\end{align}
We concentrate on the first inequality \eqref{concrete-C}.
When $\bar x_i$ 
is 0, the inequality is
trivially fulfilled.
Thus, we can restrict ourselves to the
case when $\bar x_i>0$, and hence
$\bar x_i=\hat x_i$.
 If we set $\eps=0$ and replace
$(\bar x,\bar y)$ by
$(\hat x,\hat y)$ everywhere,
the difference between the two
sides of \eqref{concrete-C} is the quantity $Q_X$ in
 Lemma~\ref{constant}, and hence it is
 bounded by $K$.
Going back from
$(\hat x,\hat y)$ to
$(\bar x,\bar y)$ cannot make the right-hand side smaller.
Thus we are done if we prove that the ``slack term''
term $\eps \cdot\bar x_i$ is at least $K$. This
is true by construction, since the non-zero values of $\bar x_i$ are
at least $K/\eps$.
The other inequality \eqref{concrete-F} follows similarly.


This concludes the proof of the lower bound.
\end{proof}

The theorem can be extended to more than two coupled recursive
sequences.
Then we need a separate parameter $\delta$ for each function in
\thetag{\ref{functions-h}--\ref{functions-h2}}
These parameters must be determined from a system of equations,
and the no-drift condition ensures that this system has a solution.

The technical condition of
Theorems~\ref{arbitrary-row-sums} and \ref{equal-row-sums}
 that certain coefficients are non-negative have the purpose to exclude
 periodicity and can be replaced by weaker conditions.


\subsection{Asymptotic growth constants}
\label{result}

We apply
Theorem~\ref{arbitrary-row-sums}
to the recursion
describing the $r$-chain with corners.
It is straightforward to compute the $2\times2$
 {condensed coefficient matrix}  \eqref{transition}
with a computer
by accumulating all terms derived
in Section~\ref{sec:corners_recursion}, and to compute
its dominant eigenvalue.
Since $n=rk+1$, the growth constant $T_r$ in terms of $n$ is $r$-th root of
this eigenvector. We observe the same phenomenon as for chains without
corners, see the right-most column of Table~\ref{tab:lambda}:
The values increase to some maximum, and then the taper off and
converge to $3$ as $r$ increases further.
The first two entries in the table reproduce the results for the
double-chain (the condensed coefficient matrix is
$\bigl(
\begin{smallmatrix}
1&1 \\ 2&2
\end{smallmatrix}
\bigr)$
which gives
$T_1=3$) and the double zigzag chain
(the condensed coefficient matrix is
$\bigl(
\begin{smallmatrix}
3&3 \\ 7&6
\end{smallmatrix}
\bigr)$
which gives
$T_2\approx  3.0532$).
We observe that the maximum is achieved for
the 8-chain with corners ($r=8$).

To establish this bound rigorously as a lower bound, we have to check the
conditions of
Theorem~\ref{arbitrary-row-sums}.
It is easy to check that the
coefficients $a^{CC}_\beta,a^{CF}_\beta,a^{FC}_\beta,a^{FF}_\beta$ are
indeed positive for
$\beta=-1,0,1$. 
The {condensed coefficient matrix}  \eqref{transition} is
\begin{equation}
\label{condensed-CF}
  \begin{pmatrix}
\bar A^{CC}
&
 \bar A^{CF} \\
 \bar A^{FC} &
 \bar A^{FF}
  \end{pmatrix}
=
\begin{pmatrix}
 2885&2619\\
 6022&  5504
\end{pmatrix}.
\end{equation}
Its dominant eigenvalue is
$ M=(8389+\sqrt{69945633}\,)/2\approx 8376.175$,
with corresponding left (unnormalized) eigenvector
 $(\rho_X,\rho_Y)= (6022,M-2885)$ and right eigenvector
 $(\pi_X,\pi_Y)^\top = (2619,M-2885)^\top$.
The matrix of total group-to-group jump sizes is
\begin{equation*}
  \begin{pmatrix}
 D^{CC}
&
 D^{CF} \\
 D^{FC} &
 D^{FF}
  \end{pmatrix}
=
\begin{pmatrix}
-2619&0\\
-2619&  2619
\end{pmatrix}.
\end{equation*}
Weighting these numbers with the eigenvectors
and summing them up~\eqref{total-jump} yields that
the {weighted total jump size} $D$ is zero.
The conditions of Theorem~\ref{arbitrary-row-sums} are thus fulfilled.

An intuitive explanation of the equality $D=0$ might be as follows.
The recursion between the two vectors $C^k$ and $F^k$ is not
symmetric, as witnessed, for example, by the non-symmetric condensed matrix
\eqref{condensed-CF}. This asymmetry comes from the arbitrary
decision to cut the construction to the \emph{right} of each corner
point.
However,
on the whole, this irregularity should not cause a systematic
``drift'' in
the recursion, which would favor a tendency towards larger or smaller
numbers $i$ of unfinished runners crossing the cut.
Thus, it is not surprising that $D=0$.
We expect that $D=0$ should hold for all $r$, but we have only checked
it numerically for small values of $r$, and we have established
it rigorously only for the concrete case $r=8$.

By Theorem~\ref{arbitrary-row-sums},
 the sequence $F^k_0$ grows at most like $M^k$ and
 at least like $(M-\eps)^k$, for any
$\eps>0$.
Since $n=8k+1$, the growth constant in terms of $n$ is
$T_8=\sqrt[8]M
\approx 3.093005695
$.
\begin{corollary}
  The $8$-chains with corners have
$O(T_8^n)$ and
$\Omega ((T_8-\eps)^n)$
down-free matchings, for every $\eps>0$.
\qed
\end{corollary}

This implies Theorem~\ref{thm:main_corners}
with the help of Theorem~\ref{thm:double}.

Numerical data
suggest the more precise estimate
$F^k_0=M^k/k^{3/2}(c_0+c_1/k+O(1/k^2))$ with $c_0\approx 0.1321$ and
$c_1\approx-0.102$.
This has been computed by
Moritz Firsching (personal communication)
by interpolation
from the elements
$F^{785}_0,F^{786}_0,F^{787}_0,\dots,F^{1000}_0$,
assuming that the sequence has the asymptotic form
$F^k_0=C^k/k^\alpha(c_0+c_¹/k+c_2/k^2+\cdots)$.
This method has predicted more than 300 correct decimal digits of
$C=M$,
and it gives $\alpha=3/2$ also to a precision of
 more than 300 digits.
By comparison,
for
the sequence  $a_k$
of down-free matching numbers of the zigzag-chain
(Section~\ref{zigzag-chains}),
for which the explicit generating function and hence the form of the asymptotic
growth is known,
 the same method gave estimates that
 were accurate also to more than 300 digits, both
regarding the growth constant
$1/\mu = (\sqrt{93}+ 9)/2$
and the power
 $\alpha=3/2$ of the polynomial factor.

The asymptotic growth of the form
$F^k_0=c_0M^k/k^{3/2}(1+o(1))$ is not unexpected; it is
in accordance with the behaviour of $r$-chains without runners, which
has been derived in the proof of
Proposition~\ref{thm:rate_wocorners}
(Section~\ref{sec:recursion-without}) by
 lattice path method~\cite[Theorem 3]{ban}.

\section{Concluding remarks}
\subsection{Table of results for $\mathsf{pm}$, $\mathsf{dfm}$ and $\mathsf{am}$}
In Table~\ref{tab:conclusion} we summarize asymptotic bounds
on different structures
for three kinds of matchings considered in this paper --
$\mathsf{pm}$, $\mathsf{dfm}$ and $\mathsf{am}$.
Some of them do not follow from results proven or mentioned in this paper,
and we explain them below.
First we want to point out some observations that can be seen in the table.

Obviously, $\mathsf{pm} (X_n) \leq \mathsf{dfm} (X_n) \leq \mathsf{am} (X_n)$,
but is $\mathsf{dfm} (X_n)$ more likely to behave similarly to $\mathsf{pm} (X_n)$ or to $\mathsf{am} (X_n)$?
Table~\ref{tab:conclusion} shows that different possibilities exist.
For a downward chain $\mathrm{SC}_n$,
every matching is down-free and thus $\mathsf{dfm} (\mathrm{SC}_n) = \mathsf{am} (X_n)$,
but for an upward chain $\mathsf{dfm}$ is equal, up to a polynomial factor, to the lower bound.
For $\mathrm{SZZC}_n$, the three growth constants are all different,
but the intermediate basis for $\mathsf{dfm}$ is closer to the upper
bound.
However for $r$-chains without corners, as $r$ grows,
the growth rate for $\mathsf{pm}$ and $\mathsf{dfm}$ tends to the same constant, $3$,
from below and from above respectively;
whereas that for $\mathsf{am}$ tends to $4$.

\begin{table}[h!]
\renewcommand{\arraystretch}{1.15}
\begin{center}
\begin{tabular}{|c||c|c|c|}
\hline
$X_n$ & $\mathsf{pm} (X_n)$ & $\mathsf{dfm} (X_n)$ & $\mathsf{am} (X_n)$  \\ \hline\hline
$\mathrm{SC}_n$ upside down & $C_{n/2}=\Theta^*(2^n)$ & $\binom{n}{\lfloor n/2 \rfloor}=\Theta^*(2^n)$ & $M_n=\Theta^*(3^n)$ \\ \hline
$\mathrm{SC}_n$
  & $C_{n/2}=\Theta^*(2^n)$ & $M_n=\Theta^*(3^n)$ & $M_n=\Theta^*(3^n)$ \\ \hline
$\mathrm{SZZC}_n$ 
& $\Theta^*(2.1974^n)$ & $\Theta^*(3.0532^n)$ & $\Theta^*(3.1022^n)$ \\ \hline
$\mathrm{CH}^*(11,n/11)$&
$\Theta^*(2.5517^n)$
& $\Theta^*(3.0840^n)$&
$\Theta^*(3.4614^n)$\\\hline
$\mathrm{CH}^*(r,n/r)$, $r \to \infty$ &
$\Theta^*(\alpha^n)$, $\alpha \nearrow 3$ &
$\Theta^*(\beta^n)$, $\beta \searrow 3$ &
$\Theta^*(\gamma^n)$, $\gamma \nearrow 4$  \\ \hline
$\mathrm{CH}(8,(n-1)/8)$&& $\Theta^*(3.0930^n)$&\\\hline
$\mathrm{CH}(r,(n-1)/r)$, $r \to \infty$ &
&
$\Theta^*(\delta^n)$, $\delta \searrow 3$ ? & \\ \hline
$\mathrm{DC}_n$ & $\Theta^*(3^n)$ & $?$ & $\Theta^*(4^n)$ \\ \hline
\end{tabular}
\end{center}
\caption{$\mathsf{pm}$, $\mathsf{dfm}$, $\mathsf{am}$ for several structures.}
\label{tab:conclusion}
\end{table}

Now we describe the entries of the table.
The first two lines are classical results, except for
the formula
$\mathsf{dfm} 
 = \binom{n}{\lfloor n/2 \rfloor}$ for an upward
chain, 
which has been proved 
in Proposition~\ref{thm:binom_gf}.

The estimate
$\mathsf{pm} (\mathrm{SZZC}_n)=\Theta^*(2.1974^n)$ from~\cite{aich1}
was mentioned in Section~\ref{sec:dzzc_pm}.
Actually, it was the fact that \textsf{pm} increases from SC to SZZC which
initially prompted us to try whether the old record of the double
structure DC could be beaten by the corresponding double structure
DZZC.
The formula $\mathsf{dfm}(\mathrm{SZZC}_n)=\Theta^*(3.0532^n)$ is
the main result of Section~\ref{SZZC}.
The estimate
$\mathsf{am} (\mathrm{SZZC}_n)=\Theta^*(3.1022^n)$ can be
derived in a similar way,
by adding
an appropriate term to the recursion~\eqref{eq:a}
for $a_k$:
the only difference is that when $P_1$ is matched to $P_3$,
the point $P_2$ can be 
free.
The singularity closest to $0$ of the resulting generating functions occurs now in
$(\sqrt{105}-9)/12$, one of the roots of
$1-9x-6x^2$.
Thus, in this case the base is $\sqrt{
\raise3pt\hbox{\vphantom X}
\raise-1pt\hbox{\vphantom X}
\smash
{12/
(\sqrt{105}-9)}} \approx
3.1022$.

 For $r$-chains without corners, $\mathrm{CH}^*(r, k)$,
the growth rate for $\mathsf{dfm}$ has been determined
in
Section~\ref{sec:r_chains_wo_analysis}, and, as was discussed
in Section~\ref{sec:recursion-without}, it
converges to 3 from above as $r\to\infty$.
The other entries in the line for
 $\mathrm{CH}^*$ can be obtained by modifying the analysis of
Section~\ref{sec:r_chains_wo_analysis};
we only need to replace appropriately
 in the formula for
$z_i^1$ in
Proposition~\ref{thm:binom_gf}
 the factor $\binom{r-i}{\lfloor
  (r-i)/2 \rfloor}$, representing the number of down-free matching
on an arc of $r-i$ points.
For
 $\mathsf{pm}$, we have to replace it by the Catalan number
 $C_{(r-i)/2}$ when $r-i$ is
even and by $0$ when $r-i$ is odd;
for $\mathsf{am}$, we replace it by the Motzkin number $M_{r-i}$.
The row sums of the recursion matrix 
can be obtained by plugging these modified expressions for
$z_i^1$ into~\eqref{eq:lambda-r}.
 For $\mathsf{pm}$,
the resulting sequence of row sums
is the sequence A189912
from~\cite{oeis},
and for $\mathsf{am}$,
it is the sequence A077587 (we omit the proofs).
From the asymptotic behavior of these sequences it follows that
their $r$-th roots, which are
the growth constants, converge to $3$ and $4$ from below.

The growth rate for \textsf{dfm} for $r$-chains with corners, $\mathrm{CH}(r, k)$, was treated in
Section~\ref{sec:with}.
Empirically, they seem to be better than $r$-chains without corners.
The monotone convergence to 3 from above is not proved.
It seems plausible that the difference between
$r$-chains with corners and $r$-chains without corners should become
negligible as $r\to\infty$, and therefore the growth constant should
converge to the same constant~3. That the convergence
should be monotonically decreasing
 is
only based on the empirical observation from
Table~\ref{tab:lambda}. We have not extended the analysis to
\textsf{pm}
and \textsf{am}, although this would be feasible with some effort. We
expect that the results would be the same as for $r$-chains without corners.

The formula
$\mathsf{pm} (\mathrm{DC}_n)=\Theta^*(3^n)$ is the classical result of
Garc\'ia, Noy, and Tejel~\cite{garcia}, in accordance with
$\mathsf{dfm} (\mathrm{SC}_n)=\Theta^*(3^n)$ from the first line.
The estimate
$\mathsf{am} (\mathrm{DC}_n)=\Theta^*(4^n)$ is due to Sharir and Welzl~\cite{sw},
and it is currently the best lower bound
on the maximum number of $\mathsf{am}$.
The growth of
$\mathsf{dfm} (\mathrm{DC}_n)$ remains unknown, but it is
$\Omega^*(3^n)$ and $O^*(4^n)$.

\medskip

We see no reason to think
that our best construction $\mathrm{CH}(8,k)$ is optimal in the
sense that it has the maximal possible $\mathsf{dfm}$
and/or
that the corresponding double construction has the maximal possible $\mathsf{pm}$.
Sets with asymptotically higher bounds 
 may very well be more
complicated --- both in terms of their description and their analysis.
An obvious continuation from single chains to $r$-chains would be to
insert a third level of downward arcs between the vertices of
$r$-chains, possibly continuing towards a fractal-like pattern. We
have not attempted to analyze these structures.

\subsection{Summary and Outlook}

We have found new constructions of point sets with a larger number of perfect
matchings than previously known.
More importantly, we show that, like for triangulations,
the true bound for perfect matchings is not given by
the double chain.
For the analysis of these sets, the notion of down-free matchings was
crucial. It allowed us to concentrate on one half of a double-construction.

We have also shown that methods from analytic combinatorics are useful for
counting problems for geometric plane graphs. However, the 
results from analytic combinatorics
{that we are aware of}
cannot be readily applied for
$r$-chains with corners. In this case, the analysis leads to coupled
recursions involving two sets of variables.
For
these recursions, we had to develop our
own methods. These somewhat pedestrian methods give the growth rate
only up to an arbitrarily small error $\eps$. We hope that that the
methods of analytic combinatorics will be further developed to
encompass such cases as well.

\paragraph*{Acknowledgements.}

Research on this paper was initiated during the
 \href{http://www.eurogiga-compose.eu/}{EuroGIGA} Final Conference
in February 2014 in Berlin.
We thank Oswin Aichholzer and Moritz Firsching for helpful discussions
and computations that supported our conjectures.
We thank Guillaume Chapuy for providing the reference to the work of
Banderier and Flajolet~\cite{ban}, and we thank Heuna Kim and Nevena
Pali\'c for suggestions to improve the presentation.

\end{document}